\documentclass[a4paper,UKenglish]{lipics}
 
\usepackage{microtype}

\usepackage{stmaryrd}
\usepackage{bussproofs}
\usepackage{xspace}
\usepackage{mdframed}
\usepackage{fancybox}
\usepackage{nicefrac}
\usepackage{xypic}
\xyoption{tips}
\SelectTips{cm}{}  
\usepackage{multicol}
\usepackage{booktabs}

\title{A Type Theory for Probabilistic and Bayesian Reasoning\footnote{This work was supported by ERC Advanced Grant QCLS: Quantum Computation, Logic and Security.}}

\author[1]{Robin Adams}
\author[1]{Bart Jacobs}
\affil[1]{Institute for Computing and Information Sciences,\\
  Radboud University, the Netherlands\\
  \texttt{\{r.adams,bart\}@cs.ru.nl}}
\authorrunning{R. Adams and B. Jacobs} 

\Copyright{Robin Adams and Bart Jacobs}

\subjclass{F.4.1 [Mathematical Logic and Formal Languages]: Mathematical Logic --- Lambda calculus and related systems; G.3 [Probability and Statistics]: Probabilistic algorithms; F.3.1 [Logics and Meanings of Programs]: Specifying and Verifying and Reasoning about Programs}
\keywords{Probability theory, type theory, effect module, Bayesian reasoning}

\serieslogo{}
\volumeinfo
  {Billy Editor and Bill Editors}
  {2}
  {Conference title on which this volume is based on}
  {1}
  {1}
  {1}
\EventShortName{}
\DOI{10.4230/LIPIcs.xxx.yyy.p}

\newcommand{\Rvar}{
 \text{(var)\xspace}
}
\newcommand{\Tvar}{
 \TBvar
 \DisplayProof
}

\newcommand{\TBvar}{
 \LeftLabel{\Rvar}
 \AxiomC{$ x : A \in \Gamma$}
 \UnaryInfC{$\Gamma \vdash x : A$}
}
\newcommand{\Rexch}{
 \text{(exch)\xspace}
}
\newcommand{\Texch}{
 \TBexch
 \DisplayProof
}

\newcommand{\TBexch}{
 \LeftLabel{\Rexch}
 \AxiomC{$\Gamma, x : A, y : B, \Delta \vdash \mathcal{J}$}
 \UnaryInfC{$\Gamma, y : B, x : A, \Delta \vdash \mathcal{J}$}
}
\newcommand{\Rref}{
 \text{(ref)\xspace}
}
\newcommand{\Tref}{
 \TBref
 \DisplayProof
}

\newcommand{\TBref}{
 \LeftLabel{\Rref}
 \AxiomC{$\Gamma \vdash t : A$}
 \UnaryInfC{$\Gamma \vdash t = t : A$}
}
\newcommand{\Rsym}{
 \text{(sym)\xspace}
}
\newcommand{\Tsym}{
 \TBsym
 \DisplayProof
}

\newcommand{\TBsym}{
 \LeftLabel{\Rsym}
 \AxiomC{$\Gamma \vdash s = t : A$}
 \UnaryInfC{$\Gamma \vdash t = s : A$}
}
\newcommand{\Rtrans}{
 \text{(trans)\xspace}
}
\newcommand{\Ttrans}{
 \TBtrans
 \DisplayProof
}

\newcommand{\TBtrans}{
 \LeftLabel{\Rtrans}
 \AxiomC{$\Gamma \vdash r = s : A$}
 \AxiomC{$\Gamma \vdash s = t : A$}
 \BinaryInfC{$\Gamma \vdash r = t : A$}
}
\newcommand{\Rmagic}{
 \text{(magic)\xspace}
}
\newcommand{\Tmagic}{
 \TBmagic
 \DisplayProof
}

\newcommand{\TBmagic}{
 \LeftLabel{\Rmagic}
 \AxiomC{$\Gamma \vdash t : 0$}
 \UnaryInfC{$\Gamma \vdash \magic{t} : A$}
}
\newcommand{\Retazero}{
 \text{($\eta 0$)\xspace}
}
\newcommand{\Tetazero}{
 \TBetazero
 \DisplayProof
}

\newcommand{\TBetazero}{
 \LeftLabel{\Retazero}
 \AxiomC{$\Gamma \vdash s : 0$}
 \AxiomC{$\Gamma \vdash t : A$}
 \BinaryInfC{$\Gamma \vdash \magic{s} = t : A$}
}
\newcommand{\Runit}{
 \text{(unit)\xspace}
}
\newcommand{\Tunit}{
 \TBunit
 \DisplayProof
}

\newcommand{\TBunit}{
 \LeftLabel{\Runit}
 \AxiomC{$$}
 \UnaryInfC{$\Gamma \vdash * : 1$}
}
\newcommand{\Retaone}{
 \text{($\eta 1$)\xspace}
}
\newcommand{\Tetaone}{
 \TBetaone
 \DisplayProof
}

\newcommand{\TBetaone}{
 \LeftLabel{\Retaone}
 \AxiomC{$\Gamma \vdash t : 1$}
 \UnaryInfC{$\Gamma \vdash t = * : 1$}
}
\newcommand{\Rinl}{
 \text{(inl)\xspace}
}
\newcommand{\Tinl}{
 \TBinl
 \DisplayProof
}

\newcommand{\TBinl}{
 \LeftLabel{\Rinl}
 \AxiomC{$\Gamma \vdash t : A$}
 \UnaryInfC{$\Gamma \vdash \inl{t} : A + B$}
}
\newcommand{\Rinleq}{
 \text{(inl-eq)\xspace}
}
\newcommand{\Tinleq}{
 \TBinleq
 \DisplayProof
}

\newcommand{\TBinleq}{
 \LeftLabel{\Rinleq}
 \AxiomC{$\Gamma \vdash t = t' : A$}
 \UnaryInfC{$\Gamma \vdash \inl{t} = \inl{t'} : A + B$}
}
\newcommand{\Rinr}{
 \text{(inr)\xspace}
}
\newcommand{\Tinr}{
 \TBinr
 \DisplayProof
}

\newcommand{\TBinr}{
 \LeftLabel{\Rinr}
 \AxiomC{$\Gamma \vdash t : B$}
 \UnaryInfC{$\Gamma \vdash \inr{t} : A + B$}
}
\newcommand{\Rinreq}{
 \text{(inr-eq)\xspace}
}
\newcommand{\Tinreq}{
 \TBinreq
 \DisplayProof
}

\newcommand{\TBinreq}{
 \LeftLabel{\Rinreq}
 \AxiomC{$\Gamma \vdash t = t' : B$}
 \UnaryInfC{$\Gamma \vdash \inr{t} = \inr{t'} : A + B$}
}
\newcommand{\Rcase}{
 \text{(case)\xspace}
}
\newcommand{\Tcase}{
 \TBcase
 \DisplayProof
}
\newcommand{\TTcase}{
 \begin{prooftree}
  \TBcase
 \end{prooftree}
}
\newcommand{\TBcase}{
 \LeftLabel{\Rcase}
 \AxiomC{$\Gamma \vdash r : A + B$}
 \AxiomC{$\Delta, x : A \vdash s : C$}
 \AxiomC{$\Delta, y : B \vdash t : C$}
 \TrinaryInfC{$\Gamma, \Delta \vdash \pcase{r}{x}{s}{y}{t} : C$}
}
\newcommand{\Rcaseeq}{
 \text{(case-eq)\xspace}
}
\newcommand{\Tcaseeq}{
 \TBcaseeq
 \DisplayProof
}

\newcommand{\TBcaseeq}{
 \LeftLabel{\Rcaseeq}
 \AxiomC{$\Gamma \vdash r = r' : A + B$}
 \AxiomC{$\Delta, x : A \vdash s = s' : C$}
 \AxiomC{$\Delta, y : B \vdash t = t': C$}
 \TrinaryInfC{$\Gamma, \Delta \vdash \pcase{r}{x}{s}{y}{t} = \pcase{r'}{x}{s'}{y}{t'} : C$}
}
\newcommand{\Rbetaplustwo}{
 \text{($\beta+_2$)\xspace}
}
\newcommand{\Tbetaplustwo}{
 \TBbetaplustwo
 \DisplayProof
}

\newcommand{\TBbetaplustwo}{
 \LeftLabel{\Rbetaplustwo}
 \AxiomC{$\Gamma \vdash r : B$}
 \AxiomC{$\Delta, x : A \vdash s : C$}
 \AxiomC{$\Delta, y : B \vdash t : C$}
 \TrinaryInfC{$\Gamma, \Delta \vdash \pcase{\inr{r}}{x}{s}{y}{t} = t[y:=r] : C$}
}
\newcommand{\Rbetaplusone}{
 \text{($\beta+_1$)\xspace}
}
\newcommand{\Tbetaplusone}{
 \TBbetaplusone
 \DisplayProof
}

\newcommand{\TBbetaplusone}{
 \LeftLabel{\Rbetaplusone}
 \AxiomC{$\Gamma \vdash r : A$}
 \AxiomC{$\Delta, x : A \vdash s : C$}
 \AxiomC{$\Delta, y : B \vdash t : C$}
 \TrinaryInfC{$\Gamma, \Delta \vdash \pcase{\inl{r}}{x}{s}{y}{t} = s[x:=r] : C$}
}
\newcommand{\Retaplus}{
 \text{($\eta+$)\xspace}
}
\newcommand{\Tetaplus}{
 \TBetaplus
 \DisplayProof
}

\newcommand{\TBetaplus}{
 \LeftLabel{\Retaplus}
 \AxiomC{$\Gamma \vdash t : A + B$}
 \UnaryInfC{$\Gamma \vdash t = \pcase{t}{x}{\inl{x}}{y}{\inr{y}} : A + B$}
}
\newcommand{\Rcasecase}{
 \text{(case-case)\xspace}
}
\newcommand{\Tcasecase}{
 \TBcasecase
 \DisplayProof
}

\newcommand{\TBcasecase}{
 \LeftLabel{\Rcasecase}
 \AxiomC{
  $\begin{array}{ccc}
    \Gamma \vdash r : A + B & \Delta, x : A \vdash s : C + D & \Delta, y : B \vdash s' : C + D\\
    \multicolumn{3}{c}{\Theta, z : C \vdash t : E \qquad \Theta, w : D \vdash t' : E}
   \end{array}$
 }
 \UnaryInfC{$\Gamma, \Delta, \Theta \vdash
\begin{array}[t]{l}
\case r \of \inl{x} \mapsto \pcase{s}{z}{t}{w}{t'} \mid \\
\qquad \inr{y} \mapsto \pcase{s'}{z}{t}{w}{t'} \\
= \case (\pcase{r}{x}{s}{y}{s'}) \\
\qquad \of \inl{z} \mapsto t \mid \inr{w} \mapsto t' : E
\end{array}$}
}
\newcommand{\Rcasepair}{
 \text{(case-$\sotimes$)\xspace}
}
\newcommand{\Tcasepair}{
 \TBcasepair
 \DisplayProof
}

\newcommand{\TBcasepair}{
 \LeftLabel{\Rcasepair}
 \AxiomC{$\Gamma \vdash r : A + B$}
 \AxiomC{$\Delta, x : A \vdash s : C$}
 \AxiomC{$\Delta, y : A \vdash s' : C$}
 \AxiomC{$\Theta \vdash t : D$}
 \QuaternaryInfC{$\Gamma, \Delta, \Theta \vdash 
\begin{array}[t]{l}
(\pcase{r}{x}{s}{y}{s'}) \sotimes t = \\
\pcase{r}{x}{s \sotimes t}{y}{s' \sotimes t} : D
\end{array}$}
}
\newcommand{\Rletcase}{
 \text{(let-case)\xspace}
}
\newcommand{\Tletcase}{
 \TBletcase
 \DisplayProof
}

\newcommand{\TBletcase}{
 \LeftLabel{\Rletcase}
 \AxiomC{
  $\begin{array}{cc}
    \Gamma \vdash r : A + B & \Delta, z : A \vdash s : C \otimes D\\
    \Delta, w : B \vdash s' : C \otimes D & \Theta, x : C, y : D \vdash t : E
   \end{array}$
 }
 \UnaryInfC{$\Gamma, \Delta, \Theta \vdash 
\begin{array}[t]{l}
\plet{x}{y}{\pcase{r}{z}{s}{w}{s'}}{t} = \\
\pcase{r}{z}{\plet{x}{y}{s}{t}}{w}{\plet{x}{y}{s'}{t}} : E
\end{array}$}
}
\newcommand{\Rinlr}{
 \text{(inlr)\xspace}
}
\newcommand{\Tinlr}{
 \TBinlr
 \DisplayProof
}
\newcommand{\TTinlr}{
 \begin{prooftree}
  \TBinlr
 \end{prooftree}
}
\newcommand{\TBinlr}{
 \LeftLabel{\Rinlr}
 \AxiomC{$\Gamma \vdash s : A + 1$}
 \AxiomC{$\Gamma \vdash t : B + 1$}
 \AxiomC{$\Gamma \vdash s \downarrow = t \uparrow : \mathbf{2}$}
 \TrinaryInfC{$\Gamma \vdash \inlr{s}{t} : A + B$}
}
\newcommand{\Rinlreq}{
 \text{(inlr-eq)\xspace}
}
\newcommand{\Tinlreq}{
 \TBinlreq
 \DisplayProof
}

\newcommand{\TBinlreq}{
 \LeftLabel{\Rinlreq}
 \AxiomC{$\Gamma \vdash s = s' : A + 1$}
 \AxiomC{$\Gamma \vdash t = t' : B + 1$}
 \AxiomC{$\Gamma \vdash s \downarrow = t \uparrow : \mathbf{2}$}
 \TrinaryInfC{$\Gamma \vdash \inlr{s}{t} = \inlr{s'}{t'} : A + B$}
}
\newcommand{\Rbetainlrone}{
 \text{($\beta$inlr$_1$)\xspace}
}
\newcommand{\Tbetainlrone}{
 \TBbetainlrone
 \DisplayProof
}

\newcommand{\TBbetainlrone}{
 \LeftLabel{\Rbetainlrone}
 \AxiomC{$\Gamma \vdash s : A + 1$}
 \AxiomC{$\Gamma \vdash t : B + 1$}
 \AxiomC{$\Gamma \vdash s \downarrow = t \uparrow : \mathbf{2}$}
 \TrinaryInfC{$\Gamma \vdash \rhd_1(\inlr{s}{t}) = s : A + 1$}
}
\newcommand{\Rbetainlrtwo}{
 \text{($\beta$inlr$_1$)\xspace}
}
\newcommand{\Tbetainlrtwo}{
 \TBbetainlrtwo
 \DisplayProof
}

\newcommand{\TBbetainlrtwo}{
 \LeftLabel{\Rbetainlrtwo}
 \AxiomC{$\Gamma \vdash s : A + 1$}
 \AxiomC{$\Gamma \vdash t : B + 1$}
 \AxiomC{$\Gamma \vdash s \downarrow = t \uparrow : \mathbf{2}$}
 \TrinaryInfC{$\Gamma \vdash \rhd_2(\inlr{s}{t}) = t : B + 1$}
}
\newcommand{\Retainlr}{
 \text{($\eta$inlr)\xspace}
}
\newcommand{\Tetainlr}{
 \TBetainlr
 \DisplayProof
}

\newcommand{\TBetainlr}{
 \LeftLabel{\Retainlr}
 \AxiomC{$\Gamma \vdash t : A + B$}
 \UnaryInfC{$\Gamma \vdash t = \inlr{\rhd_1(t)}{\rhd_2(t)} : A + B$}
}
\newcommand{\Rleft}{
 \text{(left)\xspace}
}
\newcommand{\Tleft}{
 \TBleft
 \DisplayProof
}

\newcommand{\TBleft}{
 \LeftLabel{\Rleft}
 \AxiomC{$\Gamma \vdash t : A + B$}
 \AxiomC{$\Gamma \vdash \inlprop{t} = \top : \mathbf{2}$}
 \BinaryInfC{$\Gamma \vdash \lft{t} : A$}
}
\newcommand{\Rlefteq}{
 \text{(left-eq)\xspace}
}
\newcommand{\Tlefteq}{
 \TBlefteq
 \DisplayProof
}

\newcommand{\TBlefteq}{
 \LeftLabel{\Rlefteq}
 \AxiomC{$\Gamma \vdash t = t' : A + B$}
 \AxiomC{$\Gamma \vdash \inlprop{t} = \top : \mathbf{2}$}
 \BinaryInfC{$\Gamma \vdash \lft{t} = \lft{t'} : A$}
}
\newcommand{\Rbetaleft}{
 \text{($\beta$left)\xspace}
}
\newcommand{\Tbetaleft}{
 \TBbetaleft
 \DisplayProof
}

\newcommand{\TBbetaleft}{
 \LeftLabel{\Rbetaleft}
 \AxiomC{$\Gamma \vdash t : A + B$}
 \AxiomC{$\Gamma \vdash \inlprop{t} = \top : \mathbf{2}$}
 \BinaryInfC{$\Gamma \vdash \inl{\lft{t}} = t : A + B$}
}
\newcommand{\Retaleft}{
 \text{($\eta$left)\xspace}
}
\newcommand{\Tetaleft}{
 \TBetaleft
 \DisplayProof
}

\newcommand{\TBetaleft}{
 \LeftLabel{\Retaleft}
 \AxiomC{$\Gamma \vdash t : A$}
 \UnaryInfC{$\Gamma \vdash \lft{\inl{t}} = t : A$}
}
\newcommand{\RJMprime}{
 \text{(JM)\xspace}
}

\newcommand{\TTJMprime}{
 \begin{prooftree}
  \TBJMprime
 \end{prooftree}
}
\newcommand{\TBJMprime}{
 \LeftLabel{\RJMprime}
 \AxiomC{
  $\begin{array}{cc}
    \Gamma \vdash s : (A + A) + 1 & \Gamma \vdash t : (A + A) + 1\\
    \Gamma \vdash s \goesto \rhd_1 = t \goesto \rhd_1 : A + 1 & \Gamma \vdash s \goesto \rhd_2 = t \goesto \rhd_2 : A + 1
   \end{array}$
 }
 \UnaryInfC{$\Gamma \vdash s = t : (A + A) + 1$}
}
\newcommand{\Rpair}{
 \text{($\sotimes$)\xspace}
}
\newcommand{\Tpair}{
 \TBpair
 \DisplayProof
}

\newcommand{\TBpair}{
 \LeftLabel{\Rpair}
 \AxiomC{$\Gamma \vdash s : A$}
 \AxiomC{$\Delta \vdash t : B$}
 \BinaryInfC{$\Gamma, \Delta \vdash s \sotimes t : A \otimes B$}
}
\newcommand{\Rpaireq}{
 \text{(paireq)\xspace}
}
\newcommand{\Tpaireq}{
 \TBpaireq
 \DisplayProof
}

\newcommand{\TBpaireq}{
 \LeftLabel{\Rpaireq}
 \AxiomC{$\Gamma \vdash s = s' : A$}
 \AxiomC{$\Delta \vdash t = t': B$}
 \BinaryInfC{$\Gamma, \Delta \vdash s \sotimes t = s' \sotimes t' : A \otimes B$}
}
\newcommand{\Rlett}{
 \text{(lett)\xspace}
}
\newcommand{\Tlett}{
 \TBlett
 \DisplayProof
}
\newcommand{\TTlett}{
 \begin{prooftree}
  \TBlett
 \end{prooftree}
}
\newcommand{\TBlett}{
 \LeftLabel{\Rlett}
 \AxiomC{$\Gamma \vdash s : A \otimes B$}
 \AxiomC{$\Delta, x : A, y : B \vdash t : C$}
 \BinaryInfC{$\Gamma, \Delta \vdash \plet{x}{y}{s}{t} : C$}
}
\newcommand{\Rleteq}{
 \text{(leteq)\xspace}
}
\newcommand{\Tleteq}{
 \TBleteq
 \DisplayProof
}

\newcommand{\TBleteq}{
 \LeftLabel{\Rleteq}
 \AxiomC{$\Gamma \vdash s = s' : A \otimes B$}
 \AxiomC{$\Delta, x : A, y : B \vdash t = t' : C$}
 \BinaryInfC{$\Gamma, \Delta \vdash (\plet{x}{y}{s}{t}) = (\plet{x}{y}{s'}{t'}) : C$}
}
\newcommand{\Rbeta}{
 \text{($\beta \otimes$)\xspace}
}
\newcommand{\Tbeta}{
 \TBbeta
 \DisplayProof
}

\newcommand{\TBbeta}{
 \LeftLabel{\Rbeta}
 \AxiomC{$\Gamma \vdash r : A$}
 \AxiomC{$\Delta \vdash s : B$}
 \AxiomC{$\Theta, x : A, y : B \vdash t : C$}
 \TrinaryInfC{$\Gamma, \Delta, \Theta \vdash (\plet{x}{y}{r \sotimes s}{t}) = t[x:=r,y:=s] : C$}
}
\newcommand{\Reta}{
 \text{($\eta \otimes$)\xspace}
}
\newcommand{\Teta}{
 \TBeta
 \DisplayProof
}

\newcommand{\TBeta}{
 \LeftLabel{\Reta}
 \AxiomC{$\Gamma \vdash t : A \otimes B$}
 \UnaryInfC{$\Gamma \vdash t = (\plet{x}{y}{t}{x \sotimes y}) : A \otimes B$}
}
\newcommand{\Rletlet}{
 \text{(let-let)\xspace}
}
\newcommand{\Tletlet}{
 \TBletlet
 \DisplayProof
}

\newcommand{\TBletlet}{
 \LeftLabel{\Rletlet}
 \AxiomC{$\Gamma \vdash r : A \otimes B$}
 \AxiomC{$\Delta, x : A, y : B \vdash s : C \otimes D$}
 \AxiomC{$\Theta, z : C, w : D \vdash t : E$}
 \TrinaryInfC{$\Gamma, \Delta, \Theta \vdash
\begin{array}[t]{l}
\plet{x}{y}{r}{(\plet{z}{w}{s}{t})} \\
= \plet{z}{w}{(\plet{x}{y}{r}{s})}{t}
: E
\end{array}$}
}
\newcommand{\Rletpair}{
 \text{(let-$\sotimes$)\xspace}
}
\newcommand{\Tletpair}{
 \TBletpair
 \DisplayProof
}

\newcommand{\TBletpair}{
 \LeftLabel{\Rletpair}
 \AxiomC{$\Gamma \vdash r : A \otimes B$}
 \AxiomC{$\Delta, x : A, y : B \vdash s : C$}
 \AxiomC{$\Theta \vdash t : D$}
 \TrinaryInfC{$\Gamma, \Delta, \Theta \vdash
\plet{x}{y}{r}{(s \sotimes t)} = (\plet{x}{y}{r}{s}) \sotimes t : D$}
}
\newcommand{\RleqI}{
 \text{(order)\xspace}
}

\newcommand{\TTleqI}{
 \begin{prooftree}
  \TBleqI
 \end{prooftree}
}
\newcommand{\TBleqI}{
 \LeftLabel{\RleqI}
 \AxiomC{
  $\begin{array}{cc}
    \Gamma \vdash s : A + 1 & \Gamma \vdash t : A + 1\\
    \Gamma \vdash b : (A + A) + 1 & \Gamma \vdash \doo{x}{b}{\rhd_1(x)} = s : A + 1\\
    \multicolumn{2}{c}{\Gamma \vdash \doo{x}{b}{\return \nabla(x)} = t : A + 1}
   \end{array}$
 }
 \UnaryInfC{$\Gamma \vdash s \leq t : A + 1$}
}
\newcommand{\Rinstr}{
 \text{(instr)\xspace}
}
\newcommand{\Tinstr}{
 \TBinstr
 \DisplayProof
}

\newcommand{\TBinstr}{
 \LeftLabel{\Rinstr}
 \AxiomC{$x : A \vdash t : \mathbf{n}$}
 \AxiomC{$\Gamma \vdash s : A$}
 \BinaryInfC{$\Gamma \vdash \instr_{\lambda x t}(s) : n \cdot A$}
}
\newcommand{\Rnablainstr}{
 \text{($\nabla$-instr)\xspace}
}
\newcommand{\Tnablainstr}{
 \TBnablainstr
 \DisplayProof
}

\newcommand{\TBnablainstr}{
 \LeftLabel{\Rnablainstr}
 \AxiomC{$x : A \vdash t : \mathbf{n}$}
 \AxiomC{$\Gamma \vdash s : A$}
 \BinaryInfC{$\Gamma \vdash \nabla(\instr_{\lambda x t}(s)) = s : A$}
}
\newcommand{\Rinstrtest}{
 \text{(instr-test)\xspace}
}
\newcommand{\Tinstrtest}{
 \TBinstrtest
 \DisplayProof
}

\newcommand{\TBinstrtest}{
 \LeftLabel{\Rinstrtest}
 \AxiomC{$x : A \vdash t : \mathbf{n}$}
 \AxiomC{$\Gamma \vdash s : A$}
 \BinaryInfC{$\Gamma \vdash \case_{i=1}^n \instr_{\lambda x t}(s) \of \nin{i}{n}{\_} \mapsto i = t[x:=s] : \mathbf{n}$}
}
\newcommand{\Retainstr}{
 \text{($\eta$instr)\xspace}
}
\newcommand{\Tetainstr}{
 \TBetainstr
 \DisplayProof
}

\newcommand{\TBetainstr}{
 \LeftLabel{\Retainstr}
 \AxiomC{$x : A \vdash r : n \cdot A$}
 \AxiomC{$x : A \vdash \nabla(r) = x : A$}
 \AxiomC{$\Gamma \vdash s : A$}
 \TrinaryInfC{$\Gamma \vdash \instr_{\lambda x. \case_{i=1}^n r \of \nin{i}{n}{\_} \mapsto i}(s) = r[x:=s] : n \cdot A$}
}
\newcommand{\Rinstreq}{
 \text{(instr-eq)\xspace}
}
\newcommand{\Tinstreq}{
 \TBinstreq
 \DisplayProof
}

\newcommand{\TBinstreq}{
 \LeftLabel{\Rinstreq}
 \AxiomC{$x : A \vdash t = t' : \mathbf{n}$}
 \AxiomC{$\Gamma \vdash s = s' : A$}
 \BinaryInfC{$\Gamma \vdash \instr_{\lambda x t}(s) = \instr_{\lambda x t'}(s') : n \cdot A$}
}
\newcommand{\Rcomm}{
 \text{(comm)\xspace}
}

\newcommand{\TTcomm}{
 \begin{prooftree}
  \TBcomm
 \end{prooftree}
}
\newcommand{\TBcomm}{
 \LeftLabel{\Rcomm}
 \AxiomC{$x : A \vdash p : \mathbf{2}$}
 \AxiomC{$x : A \vdash q : \mathbf{2}$}
 \AxiomC{$\Gamma \vdash t : A$}
 \TrinaryInfC{$\Gamma \vdash \begin{array}[t]{l}
\assert_{\lambda x p}(t) \goesto \assert_{\lambda x q} = \assert_{\lambda x q}(t) \goesto \assert_{\lambda x p} : A + 1
\end{array}$}
}
\newcommand{\Roneovern}{
 \text{($1 / n$)\xspace}
}
\newcommand{\Toneovern}{
 \TBoneovern
 \DisplayProof
}

\newcommand{\TBoneovern}{
 \LeftLabel{\Roneovern}
 \AxiomC{$$}
 \UnaryInfC{$\Gamma \vdash 1 / n : \mathbf{2}$}
}
\newcommand{\Rntimesoneovern}{
 \text{($n \cdot 1 / n$)\xspace}
}
\newcommand{\Tntimesoneovern}{
 \TBntimesoneovern
 \DisplayProof
}

\newcommand{\TBntimesoneovern}{
 \LeftLabel{\Rntimesoneovern}
 \AxiomC{$$}
 \UnaryInfC{$\Gamma \vdash n \cdot 1 / n = \top : \mathbf{2}$}
}
\newcommand{\Rdivide}{
 \text{(divide)\xspace}
}
\newcommand{\Tdivide}{
 \TBdivide
 \DisplayProof
}

\newcommand{\TBdivide}{
 \LeftLabel{\Rdivide}
 \AxiomC{$\Gamma \vdash n \cdot t = \top : \mathbf{2}$}
 \UnaryInfC{$\Gamma \vdash t = 1 / n : \mathbf{2}$}
}
\newcommand{\Rnorm}{
 \text{(nrm)\xspace}
}
\newcommand{\Tnorm}{
 \TBnorm
 \DisplayProof
}

\newcommand{\TBnorm}{
 \LeftLabel{\Rnorm}
 \AxiomC{$\vdash t : A + 1$}
 \AxiomC{$\vdash 1 / n \leq t : \mathbf{2}$}
 \BinaryInfC{$\Gamma \vdash \norm{t} : A$}
}
\newcommand{\Rbetanorm}{
 \text{($\beta$nrm)\xspace}
}
\newcommand{\Tbetanorm}{
 \TBbetanorm
 \DisplayProof
}

\newcommand{\TBbetanorm}{
 \LeftLabel{\Rbetanorm}
 \AxiomC{$\vdash t : A + 1$}
 \AxiomC{$\vdash 1 / n \leq t \downarrow : \mathbf{2}$}
 \BinaryInfC{$\Gamma \vdash t = \doo{\_}{t}{\return{\norm{t}}} : A + 1$}
}
\newcommand{\Retanorm}{
 \text{($\eta$nrm)\xspace}
}
\newcommand{\Tetanorm}{
 \TBetanorm
 \DisplayProof
}

\newcommand{\TBetanorm}{
 \LeftLabel{\Retanorm}
 \AxiomC{$\vdash t : A + 1$}
 \AxiomC{$\vdash 1 / n \leq t \downarrow : \mathbf{2}$}
 \AxiomC{$\vdash \rho : A$}
 \AxiomC{$\vdash t = \doo{\_}{t}{\return{\rho}} : A + 1$}
 \QuaternaryInfC{$\Gamma \vdash \rho = \norm{t} : A$}
}
\newcommand{\Rrhdoneboundmn}{
 \text{($\rhd_1-b_{mn}$)\xspace}
}
\newcommand{\Trhdoneboundmn}{
 \TBrhdoneboundmn
 \DisplayProof
}

\newcommand{\TBrhdoneboundmn}{
 \LeftLabel{\Rrhdoneboundmn}
 \AxiomC{$$}
 \RightLabel{$\left(1 \leq m < n\right)$}
 \UnaryInfC{$\Gamma \vdash \doo{x}{b_{mn}}{\rhd_1(x)} = m \cdot 1 / n : \mathbf{2}$}
}
\newcommand{\Rrhdtwoboundmnprime}{
 \text{($\rhd_2-b_{mn}$)\xspace}
}
\newcommand{\Trhdtwoboundmnprime}{
 \TBrhdtwoboundmnprime
 \DisplayProof
}

\newcommand{\TBrhdtwoboundmnprime}{
 \LeftLabel{\Rrhdtwoboundmnprime}
 \AxiomC{$$}
 \RightLabel{$\left(1 \leq m < n\right)$}
 \UnaryInfC{$\Gamma \vdash \doo{x}{b_{mn}}{\return \nabla(x)} = 1 / n : \mathbf{2}$}
}
\newcommand{\Rboundmn}{
 \text{($b_{mn}$)\xspace}
}
\newcommand{\Tboundmn}{
 \TBboundmn
 \DisplayProof
}

\newcommand{\TBboundmn}{
 \LeftLabel{\Rboundmn}
 \AxiomC{$$}
 \RightLabel{$\left(1 \leq m < n\right)$}
 \UnaryInfC{$\Gamma \vdash b_{mn} : \mathbf{3}$}
}
\newcommand{\Roveeprime}{
 \text{($\ovee$)\xspace}
}

\newcommand{\TToveeprime}{
 \begin{prooftree}
  \TBoveeprime
 \end{prooftree}
}
\newcommand{\TBoveeprime}{
 \LeftLabel{\Roveeprime}
 \AxiomC{
  $\begin{array}{cc}
    \Gamma \vdash s : A + 1 & \Gamma \vdash t : A + 1\\
    \Gamma \vdash b : (A + A) + 1 & \Gamma \vdash \doo{x}{b}{\rhd_1(x)} = s : A + 1\\
    \multicolumn{2}{c}{\Gamma \vdash \doo{x}{b}{\rhd_2(x)} = t : A + 1}
   \end{array}$
 }
 \UnaryInfC{$\Gamma \vdash s \ovee t : A + 1$}
}
\newcommand{\Roveedef}{
 \text{($\ovee$-def)\xspace}
}

\newcommand{\Sets}{\mathbf{Sets}}
\newcommand{\Kl}{\mathcal{K}{\kern-.2ex}\ell}
\newcommand{\Dst}{\mathcal{D}}
\newcommand{\Giry}{\mathcal{G}}

\newcommand{\assert}{\mathsf{assert}}
\newcommand{\case}{\mathsf{case}\ }
\newcommand{\elsen}{\ \mathsf{else}\ }
\newcommand{\idmap}[1][]{\ensuremath{\mathrm{id}_{#1}}}

\newcommand{\cond}[3]{\ifn {#1} \thenn {#2} \elsen {#3}}
\newcommand{\ifn}{\mathsf{if}\ }
\newcommand{\inln}{\mathsf{inl}}
\newcommand{\inlprop}[1]{\mathsf{inl?} \left( {#1} \right)}
\newcommand{\inlrn}{\mathsf{inlr}}
\newcommand{\inl}[1]{\inln \left( {#1} \right)}
\newcommand{\inn}{\ \mathsf{in}\ }
\newcommand{\inrn}{\mathsf{inr}}
\newcommand{\inrprop}[1]{\mathsf{inr?} \left( {#1} \right)}
\newcommand{\inr}[1]{\inrn \left( {#1} \right)}
\newcommand{\lett}{\mathsf{let}\ }
\newcommand{\lftn}{\mathsf{left}}
\newcommand{\lft}[1]{\lftn \left( {#1} \right)}
\newcommand{\instr}{\mathsf{instr}}
\newcommand{\measure}{\mathbf{TODO}}
\newcommand{\meas}{\mathsf{measure}}
\newcommand{\norm}[1]{\ensuremath{\mathsf{nrm} \left( {#1} \right)}}
\newcommand{\of}{\ \mathsf{of}\ }
\newcommand{\pcase}[5]{\case {#1} \of \inl{#2} \mapsto {#3}%
                          \mid \inr{#4} \mapsto {#5}}
\newcommand{\rgtn}{\mathsf{right}}
\newcommand{\rgt}[1]{\rgtn \left( {#1} \right)}
\newcommand{\swapper}[1]{\ensuremath{\mathsf{swap} \left( {#1} \right)}}
\newcommand{\thenn}{\ \mathsf{then}\ }

\newcommand{\nin}[3]{\mathsf{in}_{#1}^{#2} \left( {#3} \right)}
\newcommand{\return}[1]{\mathsf{return}\ {#1}}
\newcommand{\fail}{\mathsf{fail}}
\newcommand{\doo}[3]{\mathsf{do}\ {#1} \leftarrow {#2} ; {#3}}
\newcommand{\ind}[1]{\ensuremath{\mathsf{index} \left( {#1} \right)}}
\newcommand{\intest}[2]{\mathsf{in}_{#1} ? \left( {#2} \right)}
\newcommand{\condn}[2]{\mathsf{cond} \left( {#1} , {#2} \right)}
\newcommand{\bang}{\mathord{!}}

\newcommand{\inlr}[2]{\ensuremath{\text{\guillemotleft} {#1} , {#2} \text{\guillemotright}}}
\newcommand{\eqdef}{\mathrel{\smash{\stackrel{\text{def}}{=}}}}
\newcommand{\fromInit}{\,\mathop{\text{\rm \textexclamdown}}}
\newcommand{\magic}[1]{\fromInit{#1}}
\newcommand{\sotimes}{\mathrel{\raisebox{.05pc}{$\scriptstyle\otimes$}}}
\newcommand{\plet}[4]{\lett {#1} \sotimes {#2} = {#3} \inn {#4}} 
\newcommand{\slet}[3]{\lett {#1} = {#2} \inn {#3}}  
\newcommand{\ifte}[3]{\ifn {#1} \thenn {#2} \elsen {#3}}  

\newcommand{\goesto}{\ensuremath{\gg\!\!=}}
\newcommand{\andthen}{\mathrel{\&}}
\renewcommand{\ker}[1]{{#1}\!\uparrow}
\newcommand{\dom}[1]{{#1}\!\downarrow}
\newcommand{\after}{\circ}
\newcommand{\supp}{\mathop{\mathrm{supp}}}

\newcommand{\brackets}[1]{\left[ \! \left[ {#1} \right] \! \right]}
\newcommand{\Prob}[1]{\mathrm{Pr} \left( {#1} \right)}

\newcommand{\COMET}{\mathbf{COMET}}

\theoremstyle{plain}
\newtheorem{proposition}[theorem]{Proposition}

\begin{document}

\maketitle

\begin{abstract}
This paper introduces a novel type theory and logic for probabilistic
reasoning. Its logic is quantitative, with fuzzy predicates. It
includes normalisation and conditioning of states. This conditioning
uses a key aspect that distinguishes our probabilistic type theory
from quantum type theory, namely the bijective correspondence between
predicates and side-effect free actions (called instrument, or assert,
maps). The paper shows how suitable computation rules can be derived
from this predicate-action correspondence, and uses these rules for
calculating conditional probabilities in two well-known examples of
Bayesian reasoning in (graphical) models. Our type theory may thus
form the basis for a mechanisation of Bayesian inference.
\end{abstract}

\section{Introduction}

A probabilistic program is understood (semantically) as a stochastic
process.  A key feature of probabilistic programs as studied in the
1980s and 1990s is the presence of probabilistic choice, for instance
in the form of a weighted sum $x +_{r} y$, where the number $r \in
[0,1]$ determines the ratio of the contributions of $x$ and $y$ to the
result. This can be expressed explicitly as a convex sum $r\cdot x +
(1-r)\cdot y$.  Some of the relevant sources
are~\cite{Kozen81,Kozen85}, and~\cite{JonesP89},
and~\cite{MorganMS96}, and also~\cite{TixKP05} for the combination of
probability and non-determinism. In the language of category theory, a
probabilistic program is a map in the Kleisli category of the
distribution monad $\Dst$ (in the discrete case) or of the Giry monad
$\Giry$ (in the continuous case).

In recent years, with the establishement of Bayesian machine learning
as an important area of computer science, the meaning of probabilistic
programming shifted towards conditional inference. The key feature is
no longer probabilistic choice, but normalisation of distributions
(states), see \textit{e.g.}~\cite{Borgstroem2011}. Interestingly, this
can be done in basically the same underlying models, where a program
still produces a distribution --- discrete or continuous --- over its
output.

This paper contributes to this latest line of work by formulating a
novel type theory for probabilistic and Bayesian reasoning. We list
the key features of our type theory.
\begin{itemize}
\item It includes a logic, which is quantitative in nature. This means
  that its predicates are best understood as `fuzzy' predicates,
  taking values in the unit interval $[0,1]$ of probabilities, instead
  of in the two-element set $\{0,1\}$ of Booleans.

\item As a result, the predicates of this logic do not form Boolean
  algebras, but effect modules (see \emph{e.g.}~\cite{Jacobs15d}). The
  double negation rule does hold, but the sum $\ovee$ is a partial
  operation. Moreover, there is a scalar multiplication $s\cdot p$,
  for a scalar $s$ and a predicate $p$, which produces a scaled
  version of the predicate $p$.

\item This logic is a special case of a more general quantum type
  theory~\cite{Adams2014}. What we describe here is the probabilistic
  subcase of this quantum type theory, which is characterised by a
  bijective correspondence between predicates and side-effect free
  assert maps (see below for details).

\item The type theory includes normalisation (and also probabilistic
  choice). Abstractly, normalisation means that each non-zero
  `substate' in the type theory can be turned into a proper state
  (like in~\cite{JacobsWW15a}). This involves, for instance, turning a
  \emph{sub}distribution $\sum_{i}r_{i}x_{i}$, where the probabilities
  $r_{i}\in [0,1]$ satisfy $0 < r \leq 1$ for $r \eqdef
  \sum_{i}r_{i}$, into a proper distribution
  $\sum_{i}\frac{r_i}{r}x_{i}$ --- where, by construction,
  $\sum_{i}\frac{r_i}{r} = 1$.

\item The type theory also includes conditioning, via the combination
  of assert maps and normalisation (from the previous two points).
  Hence, we can calculate conditional probabilities inside the type
  theory, via appropriate (derived) computation rules. In contrast, in
  the language of~\cite{Borgstroem2011}, probabilistic (graphical)
  models can be formulated, but actual computations are done in the
  underlying mathematical models. Since these computation are done
  inside our calculus, our type theory can form the basis for
  mechanisation.
\end{itemize}

The type theory that we present is based on a new categorical
foundation for quantum logic, called effectus theory,
see~\cite{Jacobs15d,JacobsWW15a,Cho15a,ChoJWW15}\footnote{A general
  introduction to effectus theory~\cite{Cho} will soon be
  available.}. This theory involves a basic duality between states and
effects (predicates), which is implicitly also present in our type
theory. A subclass of `commutative' effectuses can be defined, forming
models for probabilistic computation and logic. Our type theory
corresponds to these commutative effectuses, and will thus be called
$\COMET$, as abbreviation of COMmutative Effectus Theory. This
$\COMET$ can be seen as an internal language for commutative
effectuses.

A key feature of quantum theory is that observations have a
side-effect: measuring a system disturbs it at the quantum level.  In
order to perform such measurements, each quantum predicate comes with
an associated `measurement' instrument operation which acts on the
underlying space. Probabilistic theories also have such instruments
\ldots but they are side-effect free!

The idea that predicates come with an associated action is familiar in
mathematics. For instance, in a Hilbert space $\mathscr{H}$, a closed
subspace $P \subseteq \mathscr{H}$ (a predicate) can equivalently be
described as a linear idempotent operator $p\colon \mathscr{H}
\rightarrow \mathscr{H}$ (an action) that has $P$ has image. We sketch
how these predicate-action correspondences also exist in the models
that underly our type theory.

First, in the category $\Sets$ of sets and functions, a predicate $p$
on a set $X$ can be identified with a subset of $X$, but also with a
`characteristic' map $p\colon X \rightarrow 1+1$, where $1+1 = 2$ is
the two-element set. We prefer the latter view. Such a predicate
corresponds bijectively to a `side-effect free' instrument
$\instr_{p} \colon X \rightarrow X+X$, namely to:
$$\begin{array}{rcl}
\instr_{p}(x)
& = &
\left\{\begin{array}{ll}
\inl{x} \mbox{\quad} & \mbox{if } p(x) = 1 \\
\inr{x} & \mbox{if } p(x) = 0 \\
\end{array}\right.
\end{array}$$

\noindent Here we write $X+X$ for the sum (coproduct), with left and
right coprojections (also called injections) $\inl{\_}, \inr{\_}
\colon X \rightarrow X+X$. Notice that this instrument merely makes a
left-right distinction, as described by the predicate, but does not
change the state $x$. It is called side-effect free because it
satisfies $\nabla \after \instr_{p} = \idmap$, where $\nabla =
[\idmap,\idmap] \colon X+X \rightarrow X$ is the codiagonal. It easy
to see that each map $f\colon X \rightarrow X+X$ with $\nabla \after f
= \idmap$ corresponds to a predicate $p\colon X \rightarrow 1+1$,
namely to $p = (\bang+\bang) \after f$, where $\bang \colon X
\rightarrow 1$ is the unique map to the final (singleton, unit) set
$1$.

Our next example describes the same predicate-action correspondence in
a probabilistic setting. It assumes familiarity with the discrete
distribution monad $\Dst$ --- see~\cite{Jacobs15d} for details, and
also Subsection~\ref{section:dpc} --- and with its Kleisli category
$\Kl(\Dst)$. A predicate map $p\colon X \rightarrow 1+1$ in
$\Kl(\Dst)$ is (essentially) a fuzzy predicate $p\colon X \rightarrow
[0,1]$, since $\Dst(1+1) = \Dst(2) \cong [0,1]$. There is also an
associated instrument map $\instr_{p} \colon X \rightarrow X+X$ in
$\Kl(\Dst)$, given by the function $\instr_{p} \colon X \rightarrow
\Dst(X+X)$ that sends an element $x\in X$ to the distribution
(formal convex combination):
$$\begin{array}{rcl}
\instr_{p}(x)
& = &
p(x)\cdot \inl{x} + (1-p(x))\cdot \inr{x}.
\end{array}$$

\noindent This instrument makes a left-right distinction, with the
weight of the distinction given by the fuzzy predicate $p$. Again we
have $\nabla \after \instr_{p} = \idmap$, in the Kleisli category,
since the instrument map does not change the state. It is easy to see
that we get a bijective correspondence.

These instrument maps $\instr_{p} \colon X \rightarrow X+X$ can in
fact be simplified further into what we call assert maps. The
(partial) map $\assert_{p} \colon X \rightarrow X+1$ can be defined as
$\assert_{p} = (\idmap+\bang) \after \instr_{p}$. We say that such a
map is side-effect free if there is an inequality $\assert_{p} \leq
\inl{\_}$, for a suitable order on the homset of partial maps $X
\rightarrow X+1$. Given assert maps for $p$, and for its
orthosupplement (negation) $p^{\bot}$, we can define the associated
instrument via a partial pairing operation as $\instr_{p} =
\inlr{\assert_p}{\assert_{p^\bot}}$, see below for details.

The key aspect of a probabilistic model, in contrast to a quantum model,
is that there is a bijective correspondence between:
\begin{itemize}
\item predicates $X \rightarrow 1+1$
\item side-effect free instruments $X \rightarrow X+X$ --- or
  equivalently, side-effect free assert maps $X \rightarrow X+1$.
\end{itemize}

\noindent We shall define conditioning via normalisation after assert.
More specifically, for a state $\omega\colon X$ and a predicate $p$ on
$X$ we define the conditional state $\omega|_{p} = \condn{\omega}{p}$
as:
$$\begin{array}{rcl}
\condn{\omega}{p}
& = &
\norm{\assert_{p}(\omega)},
\end{array}$$

\noindent where $\norm{-}$ describes normalisation (of substates to
states). This description occurs, in semantical form
in~\cite{JacobsWW15a}. Here we formalise it at a type-theoretic level
and derive suitable computation rules from it that allow us to do
(exact) conditional inference.

The paper is organised as follows. Section~\ref{section:overview}
provides an overview of the type theory, with some key results,
without giving all the details and
proofs. Section~\ref{section:examples} takes two familiar examples of
Bayesian reasoning and formalises them in our type theory $\COMET$.
Subsequently, Section~\ref{section:metatheorems} explores the type
theory in greater depth, and provides justification for the
computation rules in the examples. Next,
Section~\ref{section:semantics} sketches how our type theory can be
interpreted in set-theoretic and probabilistic
models. Appendix~\ref{section:rules} contains a formal presentation of
the type theory $\COMET$.

\section{Syntax and Rules of Deduction}
\label{section:overview}

We present here the terms and types of $\COMET$.  We shall describe the system
at a high level here, giving the intuition behind each construction.  The complete list of
the rules of deduction of $\COMET$ is given in Appendix \ref{section:rules}, and the
properties that we use are all proved in Section \ref{section:metatheorems}.

\subsection{Syntax}

Assume we are given a set of
\emph{type constants} $\mathbf{C}$, representing the base data types needed for each example.  (These may typically include for instance $\mathbf{bool}$, $\mathbf{nat}$ and $\mathbf{real}$.)
Then the types of $\COMET$ are the following.
$$ \begin{array}{lrcll}
\text{Type} & A & ::= & \mathbf{C} \mid & \text{constant type} \\
& & & 0 \mid & \text{empty type} \\
& & & 1 \mid & \text{unit type} \\
& & & A + B \mid & \text{disjoint union} \\
& & & A \otimes B & \text{pairs}
\end{array} $$

The \emph{terms} of $\COMET$ are given by the following grammar.

$$ \begin{array}{lrcll}
\text{Term} & t & ::= & x \mid & \text{variable} \\
& & & * \mid & \text{element of unit type} \\
& & & t \sotimes t \mid & \text{pair} \\
& & & \plet{x}{y}{t}{t} \mid & \text{decomposing a pair} \\
& & & \magic{t} \mid & \text{eliminate element of empty type} \\
& & & \inl{t} \mid \inr{t} \mid & \text{elements of a disjoint union} \\
& & & (\pcase{t}{x}{t}{x}{t}) \mid & \text{case distinction over union} \\
& & & \inlr{s}{t} \mid & \text{partial pairing} \\
& & & \lft{t} \mid & \text{extract element of union} \\
& & & \instr_{\lambda x t}{t} \mid & \text{instrument map} \\
& & & 1/n \mid & \text{constant scalar} (n \geq 2) \\
& & & \norm{t} \mid & \text{normalised substate} \\
& & & s \ovee t & \text{partial sum}
\end{array}$$

The variables $x$ and $y$ are bound within $s$ in $\plet{x}{y}{s}{t}$.  The variable $x$ is bound within $s$ and $y$ within $t$ in $\pcase{r}{x}{s}{y}{t}$, and $x$ is bound within $t$ in $\instr_{\lambda x t}(s)$.
We identify terms up to $\alpha$-conversion (change of bound variable).  We write $t[x:=s]$ for the result of substituting $s$ for $x$ within $t$, renaming bound variables to avoid variable capture.
We shall write $\_$ for a vacuous bound variable; for example, we write $\pcase{r}{\_}{s}{y}{t}$ for $\pcase{r}{x}{s}{y}{t}$ when $y$ does not occur free in $s$.

We shall also sometimes abbreviate our terms, for example writing $\instr_{\mathsf{inl}}(t)$ when we should strictly write $\instr_{\lambda x \inl{x}}(t)$.  Each time, the meaning should be clear from context.

The typing rules for these terms are given in Figure \ref{fig:typing}.  (Note that some of these rules make use of defined
expressions, which will be introduced in the sections below.)

\begin{figure}
\begin{mdframed}
$$ \Tvar \; \Tunit \; \Tpair $$
\TTlett
$$ \Tmagic \; \Tinl \; \Tinr $$
\TTcase
\TTinlr
$$ \Tleft \; \Tinstr $$
$$ \Toneovern \; \Tnorm $$
\TToveeprime
\end{mdframed}
\caption{Typing rules for $\COMET$}
\label{fig:typing}
\end{figure}

The typing rule for the term $\magic{t}$ 
says that from an inhabitant $t:0$ we can produce an inhabitant
$\magic{t}$ in any type $A$.  Intuitively, this says `If the empty type is inhabited,
then every type is inhabited', which is vacuously true.

A term of type $A$ is intended to represent a \emph{total} computation, that always terminates and returns a value of type $A$.
We can think of a term of type $A + 1$ as a \emph{partial} computation that may return a value $a$ of type $A$
(by outputting $\inl{a}$) or diverge (by outputting $\inr{*}$).  The judgement $s \leq t$ should be understood as:
the probability that $s$ returns $\inl{a}$ is $\leq$ the probability that $t$ returns $\inl{a}$, for all $a$.  The rule for this
ordering relation is given in Figure \ref{fig:ordering}.

\begin{figure}
\begin{mdframed}
\TTleqI
\end{mdframed}
\caption{Rule for Ordering in $\COMET$}
\label{fig:ordering}
\end{figure}

The term $\inlr{s}{t}$ is understood intuitively as follows.  We are
given two partial computations $s$ and $t$, and we have derived the
judgement $\dom{s} = \ker{t}$, which tells us that exactly one of
$s$ and $t$ converges on any given input.  We may then form the
computation $\inlr{s}{t}$ which, given an input $x$, returns either
$s(x)$ or $t(x)$, whichever of the two converges.

For the term $\lft{t}$: if we have a term $t : A + B$ and we have derived the judgement $\inlprop{t} = \top$, then we know
that $t$ has the form $\inl{a}$ for some term $a : A$.  We denote this unique term $a$ by $\lft{t}$.

For the term $\instr_{\lambda x t}(s)$: think of the type $\mathbf{n}$ as the set $\{ 1, \ldots, n \}$.  The elements of the type $A + \cdots + A$ consist of $n$ copies of each element $a$ of $A$,
denoted $\nin{1}{n}{a}$, \ldots, $\nin{n}{n}{a}$.  Then $\instr_{\lambda x t}(s)$ is the object $\nin{t[x:=s]}{n}{s}$.  It maps $s$ into one of the $n$ copies of $A$, which one being
determined by the test $t$.

The term $1 / n$ represents the probability distribution on $\mathbf{2} = \{ \top, \bot \}$ which returns $\top$ with probability $1 / n$ and $\bot$ with probability $(n - 1) / n$.  It can
be thought of as a coin toss, with a weighted coin that returns heads with probability $1 / n$.

For the term $\norm{t}$: the term $t : A + 1$ represents a distribution on $A + 1$.  Let $s$ denote the probability that $t$ terminates (i.e. returns a term of the form $\inl{a}$), and let
$\omega(a)$ denote the probability that $t$ returns $a$.  Then $\norm{t}$ returns $a$ with probability $\omega(a) / s$.  Thus, $\norm{t}$ is the distribution resulting from normalising
the subdistribution given by $t$.

The term $s \ovee t$ is the `sum' of $s$ and $t$ in the following sense.  It is defined on a given input if and only if, for any $a$, the probability that $s$ and $t$ both return $\inl{a}$ is $\leq 1$.
In this case, the probability that $s \ovee t$ returns $\inl{a}$ is the sum of these two probabilities.

The computation rules that these terms obey are given in Figure \ref{fig:equations}.

\begin{figure}
\begin{mdframed}
\begin{gather*}
\plet{x}{y}{r \sotimes s}{t}  = t[x:=r,y:=s] \tag*{\Rbeta} \\
\pcase{\inl{r}}{x}{s}{y}{t}  = s[x:=r]  \tag*{\Rbetaplusone} \\
\pcase{\inr{r}}{x}{s}{y}{t}  = t[y:=r]  \tag*{\Rbetaplustwo} \\
\rhd_1(\inlr{s}{t})  = s  \tag*{\Rbetainlrone} \\
\rhd_2(\inlr{s}{t})  = t  \tag*{\Rbetainlrtwo} \\
\inl{\lft{t}}  = t  \tag*{\Rbetaleft} \\
\lft{\inl{t}}  = t  \tag*{\Retaleft} \\
\ind{\instr_{\lambda x p}(t)}  = p[x:=t]  \tag*{\Rinstrtest} \\
\nabla(\instr_{\lambda x p}(t))  = t  \tag*{\Rnablainstr} \\
\text{if } \nabla(t) = x \text{ then } \instr_{\lambda x \ind{t}}(s)  = t[x:=s]  \tag*{\Retainstr} \\
\text{if } t : 1 \text{ then }   *  = t  \tag*{\Retaone} \\
\text{if } t : A \otimes B \text{ then }   \plet{x}{y}{t}{x \sotimes y} = t  \tag*{\Reta} \\
\text{if } t : A + B \text{ then }   t\case t \of \inl{x} \mapsto \inl{x} \mid \inr{y} \mapsto \inr{y} = t \tag*{\Retaplus} \\
\text{if } t : A + B \text{ then }   \inlr{\rhd_1(t)}{\rhd_2(t)} = t \tag*{\Retainlr} \\
\text{if } t \text{ is well-typed then }   \doo{\_}{t}{\return{\norm{t}}} = t  \tag*{\Rbetanorm} \\
\text{if } t = \doo{\_}{t}{\return{\rho}} \text{ and } 1 / n \leq t, \text{ then } \rho = \norm{t}
\tag*{\Retanorm} \\
n \cdot 1 / n = \top
 \tag*{\Rntimesoneovern} \\
\text{if } n \cdot t = \top \text{ then } t = 1/n
 \tag*{\Rdivide}
\end{gather*}
\end{mdframed}
\caption{Computation rules for $\COMET$}
\label{fig:equations}
\end{figure}

Figures \ref{fig:typing} and \ref{fig:equations} should be understood
simultaneously.  So the term $\inlr{s}{t}$ is well-typed if and only
if we can type $s : A + 1$ and $t : B + 1$ (using the rules in Figure
\ref{fig:typing}), \emph{and} derive the equation $\dom{s} = \ker{t}$
using the rules in Figure~\ref{fig:equations}.

The full set of rules of deduction for the system is given in Appendix \ref{section:rules}.

\subsection{Linear Type Theory}

Note the form of several of the typing rules in Figure \ref{fig:typing}, including\Rpair and\Rlett.  These rules do not allow
a variable to be duplicated; in particular, we cannot derive the judgement $x : A \vdash x \sotimes x : A \otimes A$.  The \emph{contraction} rule does not hold in our type theory --- it is not the case in general that, if $\Gamma, x : A, y : B \vdash \mathcal{J}$, then $\Gamma, z : A \vdash \mathcal{J}[x:=z,y:=z]$.  Our theory is thus similar to a \emph{linear} type theory (see for example \cite{Benton93aterm}).

The reason is that these judgements do not behave well with respect to substitution.  For example, take the computation $x : \mathbf{2} \vdash x \sotimes x : 2 \otimes 2$.
If we apply this computation to the scalar $1 / 2$, we presumably wish the result to be $\top \sotimes \top$ with probability $1/2$, and $\bot \sotimes \bot$ with probability $1/2$.  But
this is not the semantics for the term $\vdash 1/2 \sotimes 1/2 : 2 \otimes 2$.  This term assigns probability $1/4$ to all four possibilities $\top \sotimes \top$, $\top \sotimes \bot$, $\bot \sotimes \top$,
$\top \sotimes \top$.

\subsection{Defined Constructions}
We can define the following types and computations from the primitive constructions given above.

\subsubsection{States, Predicates and Scalars}

A closed term $\vdash t : A$ will be called a \emph{state} of type $A$, and intuitively it represents a probability distribution over the elements of $A$.

A \emph{predicate} on type $A$ is a proposition of the form $x : A \vdash p : \mathbf{2}$.  These shall be the formulas of the logic of $\COMET$ (see Section \ref{section:logic}).

A \emph{scalar} is a term $s$ such that $\vdash s : \mathbf{2}$.
The closed terms $t$ such that $\vdash t : \mathbf{2}$ are called \emph{scalars}, and represent the \emph{probabilities} or \emph{truth values} of our system.  In our intended semantics for discrete and continuous probabilities, these
denote elements of the real interval $[0,1]$.

Given a state $\vdash t : A$ and a predicate $x : A \vdash p : \mathbf{2}$, we can find the probability that $p$ is true when measured on $t$; this probability is simply the scalar $p[x:=t]$.

\subsubsection{Coproducts and Copowers}
\label{section:copowers}

Since we have the coproduct $A + B$ of two types, we can construct the disjoint union of $n$ types $A_1 + \cdots + A_n$ in the obvious way.  We write $\nin{1}{n}{}$, \ldots, $\nin{n}{n}{}$
for its constructors; thus, if $a : A_i$ then $\nin{i}{n}{a} : A_1 + \cdots + A_n$.  And given $t : A_1 + \cdots + A_n$, we can eliminate it as:
\[ \case t \of \nin{1}{n}{x_1} \mapsto t_1 \mid \cdots \mid \nin{n}{n}{x_n} \mapsto t_n \enspace . \]
We abbreviate this expression as $\case_{i=1}^n\ t \of \nin{i}{n}{x_i} \mapsto t_i$.

For the special case where all the types are equal, we write $n \cdot A$ for the type $A + \cdots + A$, where there are $n$ copies of $A$.  In category
theory, this is known as the $n$th \emph{copower} of $A$.  (We
include the special cases $0 \cdot A \eqdef 0$ and $1 \cdot A \eqdef A$.)

The \emph{codiagonal} $\nabla(t) : A$ for $t : n \cdot A$ is defined by
\[ \nabla(t) = \case_{i=1}^n\ t \of \nin{i}{n}{x} \mapsto x \enspace . \]
This computation extracts the value of type $A$ and discards the information about which of the $n$ copies it came from.

We write $\mathbf{n}$ for $n \cdot 1$.  Intuitively, this is a finite type with $n$ canonical elements.  We denote these elements by $1$, $2$, \ldots, $n$:
\[ i \eqdef \nin{i}{n}{*} : \mathbf{n} \qquad (1 \leq i \leq n) \enspace . \]
For $t : n \cdot A$, we define
\[ \ind{t} = \case_{i=1}^n t \of \nin{i}{n}{\_} \mapsto i : \mathbf{n} \enspace . \]
Thus, if $t = \nin{i}{n}{a}$, then $\ind{t}$ extracts the index $i$ and throws away the value $a$.

We have the $\lft{}$ construction, which extracts a term of type $A$ from a term of type $A + B$.
We have a similar $\rgt{}$ construction, but there is no need to give primitive rules for this one, as it can be defined in terms of $\lft{}$:
\[ \rgt{t} \eqdef \lft{\swapper{t}} \]
where $\swapper{t} = \pcase{t}{x}{\inr{x}}{y}{\inl{y}}$.

\subsubsection{Partial Functions}

We may see a term $\Gamma \vdash t : A + 1$ as denoting a \emph{partial function} into $A$, which has some probability of terminating (returning a value of form $\inl{s}$) and some probability of diverging (returning $\inr{*}$).
We shall introduce the following notation for dealing with partial functions.

We define:
\begin{itemize}
\item If $\Gamma \vdash t : A$ then $\Gamma \vdash \return{t} \eqdef \inl{t} : A + 1$.  This program converges with probability 1.
\item $\Gamma \vdash \fail \eqdef \inr{*} : A + 1$.  This program diverges with probability 1.
\item If $\Gamma \vdash s : A + 1$ and $\Delta, x : A \vdash t : B + 1$ then \\
$\Gamma, \Delta \vdash \doo{x}{s}{t} \eqdef \pcase{s}{x}{t}{\_}{\fail}$.
\item We introduce the following abbreviation.  If $f$ is an expression (such as $\inln$, $\inrn$) such that $f(x)$ is a term, then we write $t \goesto f$ for $\doo{x}{t}{f(x)}$.
\end{itemize}

The term $\doo{x}{s}{t}$ should be read as the following computation: Run $s$.  If $s$ returns a value, pass this as input $x$ to the computation $t$; otherwise, diverge.

These constructions satisfy these computation rules (Lemma \ref{lm:do}):
\begin{align*}
\doo{x}{\return{s}}{t} & = t[x:=s] \\
\doo{x}{\fail}{t} & = \fail \\
\doo{x}{r}{\return{x}} & = r \\
\doo{\_}{r}{\fail} & = \fail \\
\doo{x}{r}{(\doo{y}{s}{t})} & = \doo{y}{(\doo{x}{r}{s})}{t} 
\end{align*}

This construction also allows us to define \emph{scalar multiplication}.  Given a scalar $\vdash s : \mathbf{2}$ and a substate $\vdash t : A + 1$, the result of multiplying or scaling $t$ by $s$ is $\vdash \doo{\_}{s}{t} : A + 1$.

\paragraph{Partial Projections}

Recall that $n \cdot A$ has, as objects, $n$ copies of each object $a : A$, namely $\nin{1}{n}{a}$, \ldots, $\nin{n}{n}{a}$.
Given $t : n \cdot A$, the \emph{partial projection} $\rhd_{i_1 i_2 \cdots i_k}^{n}(t) : A + 1$ is the partial computation that:
\begin{itemize}
\item given an element $\nin{i_r}{n}{a}$, returns $a$;
\item given an element $\nin{j}{n}{a}$ for $j \neq i_1, \ldots, i_k$, diverges.
\end{itemize}

Formally, we define
\[ \rhd_{i_1 i_2 \cdots i_k}^{n}(t) \eqdef \case_{i=1}^n t \of \nin{i}{n}{x} \mapsto \begin{cases}
\return{x} & \text{if}\ i = i_1, \ldots, i_k \\
\fail & \text{otherwise}
\end{cases} \]

\paragraph{Partial Sum}
\label{section:ordering}

Let $\Gamma \vdash s,t : A + 1$.  If these have disjoint domains (i.e. given any input $x$, the sum of the probability that $s$ and $t$ return $a$ is never greater than 1), then we may form the computation $\Gamma \vdash s \ovee t$,
the \emph{partial sum} of $s$ and $t$.  The probability that this program converges with output $a$ is the sum of the probability that $s$ returns $a$, and the probability that $t$ returns $a$.  The definition is given by the rule \Roveedef; see Section \ref{section:psum}.

We write $n \cdot t$ for the sum $t \ovee \cdots \ovee t$ with $n$ summands.  (We include the special cases $0 \cdot t = \fail$ and $1 \cdot t = t$.)

With this operation, the partial functions
in $A + 1$ form a \emph{partial commutative monoid} (PCM) (see Lemma \ref{lm:ordering}).

\subsection{Logic}
\label{section:logic}

The type $\mathbf{2} = 1 + 1$ shall play a special role in this type theory.  It is the type of \emph{propositions} or \emph{predicates}, and its objects shall be used as the formulas of our logic.

We define $\top \eqdef \inl{*}$ and $\bot \eqdef \inr{*}$.  
We also define the \emph{orthosupplement} of a predicate $p$, which roughly corresponds to negation:
\[ p^\bot \eqdef \pcase{p}{\_}{\bot}{\_}{\top} \]

We immediately have that $p^{\bot \bot} = p$, $\top^\bot = \bot$ and $\bot^\bot = \top$.

The ordering on $\mathbf{2}$ shall play the role of the \emph{derivability} relation in our logic: $p \leq q$ will indicate that $q$ is derivable from $p$, or that $p$ implies $q$.  The rules for this logic
are not the familiar rules of classical or intuitionistic logic.  Rather, the predicates over any context form an \emph{effect algebra} (Proposition \ref{prop:logic}).

In the case of two predicates $p$ and $q$, the partial sum can be thought of as the proposition `$p$ or $q$'.  However,
it differs from disjunction in classical or intuitionistic logic as it is a \emph{partial} operation: it is only defined if $p \leq q^\bot$ (Proposition \ref{prop:logic}.\ref{prop:ortho}).
This condition can be thought of as expressing that $s$ and $t$ are \emph{disjoint}; that is, they are never both true.  

\subsubsection{$n$-tests}

An \emph{$n$-test} in a context $\Gamma$ is an $n$-tuple of predicates $(p_1, \ldots, p_n)$ on $A$ such that
\[ \Gamma \vdash p_1 \ovee \cdots \ovee p_n = \top : \mathbf{2} \enspace . \]

Intutively, this can be thought of as a set of $n$ fuzzy predicates whose probabilities always sum to 1.  We can think of this as a test that
can be performed on the types of $\Gamma$ with $n$ possible outcomes; and, indeed, there is a one-to-one correspondence between
the $n$-tests of $\Gamma$ and the terms of type $\mathbf{n}$ (Lemma \ref{lm:ntest}).

\subsubsection{Instrument Maps}

Let $x : A \vdash t : \mathbf{n}$ and $\Gamma \vdash s : A$.
The term $\instr_{\lambda x t}(s) : n \cdot A$ is interpreted as follows: we read the computation $x : A \vdash t : \mathbf{n}$ as a test on the type $A$, with $n$ possible outcomes.
The computation $\instr_{\lambda x t}(s)$ runs $t$ on (the output of) $s$, and returns either $\nin{i}{n}{s}$, where $i$ is the outcome of the test.


Given an $n$-test $(p_1, \ldots, p_n)$ on $A$, we can write a program that tests which of $p_1$, \ldots, $p_n$ is true of its input, and performs one of $n$ different calculations
as a result.  We write this program as
\[ \Gamma \vdash \mathsf{measure}\ p_1 \mapsto t_1 \mid \cdots \mid p_n \mapsto t_n \enspace . \]
It will be defined in Definition \ref{df:measure}.

If $x : A \vdash p : \mathbf{2}$ and $\Gamma, x : A \vdash s,t : A$, we define
\[ \Gamma \vdash (\cond{p}{s}{t}) = \meas\ p \mapsto s \mid p^\bot \mapsto t \enspace . \]
In the case where $s$ and $t$ do not depend on $x$,
we have the following fact (Lemma \ref{lm:measuretwo}.\ref{lm:measurecond}):
\[ \cond{p}{s}{t} = \pcase{p}{\_}{s}{\_}{t} \]

\subsubsection{Assert Maps}

If $x : A \vdash p : \mathbf{2}$ is a predicate, we define
\[ \Gamma \vdash \assert_{\lambda x p}(t) \eqdef \pcase{\instr_{\lambda x p}(t)}{x}{\return{x}}{\_}{\fail} : A + 1 \]
The computation $\assert_p(t)$ is a partial computation with output type $A$.  It tests whether $p$ is true of $t$; if so, it leaves $t$ unchanged; if not, it diverges.
That is, if $p[x:=t]$ returns $\top$, the computation converges and returns $t$; if not, it diverges.

These constructions satisfy the following computation rules (see Section \ref{section:assert} below for the proofs).

\newcommand{\Rassertdown}{(assert$\downarrow$)\xspace}
\newcommand{\Rassertscalar}{(assert-scalar)\xspace}
\newcommand{\Rinstrplus}{(instr$+$)\xspace}
\newcommand{\Rassertplus}{(assert$+$)\xspace}
\newcommand{\Rinstrm}{(instr $m$)\xspace}
\newcommand{\Rassertm}{(assert $m$)\xspace}
\begin{description}
\item[\Rassertdown]
$\dom{(\assert_{\lambda x p}(t))} = p[x:=t]$
\item[\Rassertscalar]
For a scalar $\vdash s : \mathbf{2}$: $\assert_{\lambda \_ s}(*) = \instr_{\lambda \_ s}(*) = s : \mathbf{2}$.
\item[\Rinstrplus]
For $x : A + B \vdash t : \mathbf{n}$:
\begin{align*}
\instr_{\lambda x t}(s) =
\case s \of
& \inl{y} \mapsto \case_{i=1}^n \instr_{\lambda a. t[x:=\inl{a}]}(y) \of \nin{i}{n}{z} \mapsto \nin{i}{n}{\inl{z}} \\
& \inr{y} \mapsto \case_{i=1}^n \instr_{\lambda b.t[x:=\inl{b}]}(y) \of \nin{i}{n}{z} \mapsto \nin{i}{n}{\inr{z}} \\
\end{align*}
\item[\Rassertplus]
For $x : A + B \vdash p : \mathbf{2}$:
\begin{align*}
\assert_{\lambda x p}(t) = \case t \of
& \inl{x} \mapsto \doo{z}{\assert_{\lambda a. p[x:=\inl{a}]}(x)}{\return{\inl{z}}} \mid \\
& \inr{y} \mapsto \doo{z}{\assert_{\lambda b.p[x:=\inr{b}]}(y)}{\return{\inr{z}}}
\end{align*}
\item[\Rinstrm]
For $x : \mathbf{m} \vdash t : \mathbf{n}$:
\[ \instr_{\lambda x t}(s) = \case_{i=1}^m s \of i \mapsto \case_{j=1}^n t[x:=i] \of j \mapsto \nin{j}{n}{i} \]
\item[\Rassertm]
For $x : \mathbf{m} \vdash p : \mathbf{2}$:
\[ \assert_{\lambda x p}(t) = \case_{i=1}^m t \of i \mapsto \cond{p[x:=i]}{\return{i}}{\fail} \]
\end{description}

In particular, we have $\assert_{\inln?}(t) = \rhd_1(t)$ and $\assert_{\inrn?}(t) = \rhd_2(t)$.

\subsubsection{Sequential Product}

Given two predicates $x : A \vdash p,q : \mathbf{2}$, we can define their \emph{sequential product}
\[ x : A \vdash p \andthen q \eqdef \doo{x}{\assert_p(x)}{q} : \mathbf{2} \enspace . \]
The probability of this predicate being true at $x$ is the product of the probabilities of $p$ and $q$.
This operation has many of the familiar properties of conjunction --- including commutativity --- but not all: in particular, we do not have $p \andthen p^\bot = \bot$ in all cases.  (For example, $1/2 \andthen (1 / 2)^\bot = 1/4$.)

\subsubsection{Coproducts}

We can define predicates which, given a term $t : A + B$, test which of $A$ and $B$ the term came from.  We write these as $\inlprop{t}$ and $\inrprop{t}$.  (Compare these with the operators
$FstAnd$ and $SndAnd$ defined in \cite{Jacobs14}.)  They are defined by

\begin{align*}
\inlprop{t} & \eqdef \pcase{t}{\_}{\top}{\_}{\bot} \\
\inrprop{t} & \eqdef \pcase{t}{\_}{\bot}{\_}{\top}
\end{align*}

\subsubsection{Kernels}

The predicate $\inrprop{}$ is particularly important for partial maps.

Let $\Gamma \vdash t : A + 1$.  The \emph{kernel} of the map denoted by $t$ is
\[ \ker{t} \eqdef \inrprop{t} \eqdef \pcase{t}{\_}{\bot}{\_}{\top} \]
Intuitively, if we think of $t$ as a partial computation, then
$\ker{t}$ is the proposition `$t$ does not terminate', or the function
that gives the probability that $t$ will diverge on a given input.

Its orthosupplement, $(\ker{t})^\bot = \inlprop{t}$, which we shall
also write as $\dom{t}$, is also called the \emph{domain
  predicate} of $t$, and represents the proposition that $t$
terminates.  We note that it is equal to $\doo{\_}{t}{\top}$.







\subsubsection{Normalisation}\label{subsec:normalisation}

We have a representation of all the rational numbers in our system: let $m / n$ be the term
\[ \overbrace{1 / n \ovee \cdots \ovee 1 / n}^{m} \enspace . \]
The usual arithmetic of rational numbers (between 0 and 1) can be carried out in our system (see Section \ref{sec:scalars}).  In particular, for rational numbers $q$ and $r$, we have that if $q \leq r$ then the judgement $q \leq r$ is derivable;
$q \ovee r$ is well-typed if and only if $q + r \leq 1$, in which case $q \ovee r$ is equal to $q + r$; and $q \andthen r = qr$.

Now, let $\vdash t : A + 1$.  Then $t$ represents a \emph{substate} of $A$.
As long as the probability $\dom{t}$ is non-zero, we can
\emph{normalise} this program over the probability of non-termination.  The result is the state denoted by $\norm{t}$.  Intuitively, the probability that $\norm{t}$ will output $a$ is the probability
that $t$ will output $\inl{a}$, conditioned on the event that $t$ terminates.

In order to type $\norm{t}$, we must first prove that $t$ has a
non-zero probability of terminating by deriving an inequality of the
form $1 / n \leq \dom{t}$ for some positive integer $n \geq 2$.

If $\vdash t : A$ and $x : A \vdash p(x) : \mathbf{2}$, we write $\condn{t}{p}$ for
\[ \condn{t}{p} \eqdef \norm{\assert_p(t)} \enspace . \]
The term $t$ denotes a computation whose output is given by a probability distribution over $A$.  Then $\condn{t}{p}$ gives the result of normalising that conditional probability distribution
with respect to $p$.

\subsubsection{Marginalisation}

The tensor product of type $A \otimes B$ comes with two \emph{projections}.  Given $\Gamma \vdash t : A \otimes B$, define
\begin{align*}
  \Gamma \vdash \pi_1(t) \eqdef \plet{x}{\_}{t}{x} : A \\
\Gamma \vdash \pi_2(t) \eqdef \plet{\_}{y}{t}{y} : B
\end{align*}
If $t$ is a state (i..e $\Gamma$ is the empty context), then $\pi_1(t)$ denotes the result of \emph{marginalising} $t$, as
a probability distribution over $A \otimes B$, to a probability distribution over $A$.

\subsubsection{Local Definition}

In our examples, we shall make free use of \emph{local definition}.  This is not a part of the syntax of $\COMET$ itself,
but part of our metalanguage.  We write $\lett x = s \inn t$ for $t[x:=s]$.  We shall also locally define functions: we
write $\lett f(x) = s \inn t$ for the result of replacing every subterm of the form $f(r)$ with $s[x:=r]$ in $t$.

\section{Examples}
\label{section:examples}

This section describes two examples of (Bayesian) reasoning in our
type theory $\COMET$. The first example is a typical exercise in
Bayesian probability theory. Since such kind of reasoning is not very
intuitive, a formal calculus is very useful. The second example
involves a simple graphical model.

\begin{example}
  (See also \cite{Yudkowsky2003,Borgstroem2011}) Consider the
  following situation.
  \begin{quote}
    1\% of a population have a disease.  80\% of subjects with the
    disease test positive, and 9.6\% without the disease also test
    positive.  If a subject is positive, what are the odds they have
    the disease?
  \end{quote}

\noindent This situation can be described as a very simple graphical
model, with associated (conditional) probabilities.
$$\vcenter
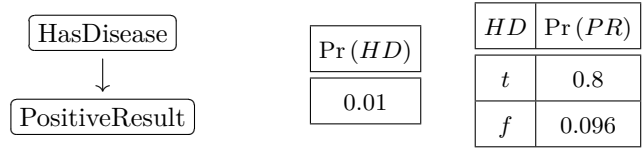
{\xymatrix@R-1pc@C-2pc{
\ovalbox{HasDisease}\ar[d]
\\
\ovalbox{PositiveResult}
}}
\qquad\qquad
{\setlength\tabcolsep{0.3em}\begin{tabular}{|c|}
\hline
$\Prob{HD}$ \\
\hline\hline
$0.01$ \\
\hline
\end{tabular}}
\qquad
{\setlength\tabcolsep{0.3em}\begin{tabular}{|c|c|}
\hline
$HD$ & $\Prob{PR}$ \\
\hline\hline
$t$ & $0.8$ \\
\hline
$f$ & $0.096$ \\
\hline
\end{tabular}}$$

\newcommand{\subj}{\textsf{subject}}
\newcommand{\pr}{\textsf{positive\_result}}

\noindent In our type theory $\COMET$, we use the following description.
$$\begin{array}{l}
\slet{\subj}{0.01}{} \\
\qquad\slet{\pr(x)}{(\ifte{x}{0.8}{0.096})}{} \\
\condn{\subj}{\pr}
\end{array}$$

\noindent We thus obtain a state $\subj : \mathbf{2}$,
conditioned on the predicate $\pr$ on $\mathbf{2}$. We calculate the
outcome in semi-formal style. The conditional state
$\condn{\subj}{\pr}$ is defined via normalisation of assert, see
Subsection~\ref{subsec:normalisation}.  
We first show what this assert
term is, using the rule \Rassertm and \Rassertscalar:
$$\begin{array}{rcl}
\assert_{\pr}(x)
& = & 
\ifn\ x \begin{array}[t]{ll}
\thenn & \doo{\_}{\assert_{\pr(\top)}(x)}{\return{\top}} \\
\elsen & \doo{\_}{\assert_{\pr(\bot)}(x)}{\return{\bot}}
\end{array} \\
& = & 
\ifn\ x \begin{array}[t]{ll}
\thenn & \doo{\_}{\assert_{0.8}(x)}{\return{\top}} \\
\elsen & \doo{\_}{\assert_{0.096}(x)}{\return{\bot}}
\end{array} \\
& = & \ifn x {\begin{array}[t]{rl}
\thenn & \cond{0.8}{\return{\top}}{\fail} \\
\elsen & \cond{0.096}{\return{\bot}}{\fail}
\end{array}}
\end{array}$$
\noindent Conditioning requires that the domain of the substate
$\assert_{\pr}(\subj)$ is non-zero. We compute this domain as:
$$\begin{array}{rcll}
\dom{\assert_{\pr}(\subj)}
& = &
\pr(\subj) & (\text{Rule \Rassertdown}) \\
& = &
\cond{0.01}{0.8}{0.096} \\
& = &
0.01 \andthen 0.8 \ovee 0.99 \andthen 0.096 \mbox{\qquad} & 
   (\text{Lemma \ref{lm:measuretwo}.\ref{lm:measurecond}}) \\
& = &
0.10304 & (\text{Lemma \ref{lm:rational}})
\end{array}$$

\noindent Hence we can choose (for example) $n = 10$, to get
$\frac{1}{n} \leq 0.10304 = \dom{\assert_{\pr}(\subj)}$.

We now proceed to calculate the result, answering the question in
the beginning of this example.
$$\begin{array}{rcll}
\assert_{\pr}(\subj) 
& = & 
\ifn 0.01 {\begin{array}[t]{rl}
\thenn & \cond{0.8}{\return{\top}}{\fail} \\
\elsen & \cond{0.096}{\return{\bot}}{\fail}
\end{array}} \\
& = & 
\meas\ {\begin{array}[t]{lcl}
0.01 \andthen 0.8 
& \mapsto &
\return{\top} \\
0.01 \andthen 0.8^\bot 
& \mapsto &
\fail \\
0.01^\bot \andthen 0.096 
& \mapsto &
\return{\bot} \\
0.01^\bot \andthen 0.096^\bot 
& \mapsto 
& \fail 
\end{array}} & (\text{Lemma \ref{lm:measure}.\ref{lm:measureand}}) 
\\
& = &  
\meas\ {\begin{array}[t]{lcl}
0.008 
& \mapsto &
\return{\top} \\
0.09504 
& \mapsto &
\return{\bot} \\
0.89696 
& \mapsto &
\fail
\end{array}} & (\text{Lemma \ref{lm:measure}.\ref{lm:measureor}}) 
\\
\condn{\subj}{\pr}
& \eqdef &
\norm{\assert_{\pr}(\subj)} \\
& = & 
\meas {\begin{array}[t]{lcl}
0.0776 
& \mapsto &
\top \\
0.9224 
& \mapsto &
\bot 
\end{array}} & (\text{Corollary \ref{cor:normmeasure}}) \\
& = & 
0.0776. & (\text{Lemma \ref{lm:measuretwo}.\ref{lm:measuretwo'}})
\end{array}$$

\noindent Hence the probability of having the disease after a positive
test result is 7.8\%.

\end{example}

\begin{example}[Bayesian Network]
The following is a standard example of a problem in Bayesian networks,
created by~\cite[Chap.~14]{RusselN03}.


\begin{quote}
I’m at work, neighbor John calls to say my alarm is ringing. Sometimes
it’s set off by minor earthquakes. Is there a burglar?
\end{quote}

We are given that the situation is as described by the following
Bayesian network.
$$\vcenter
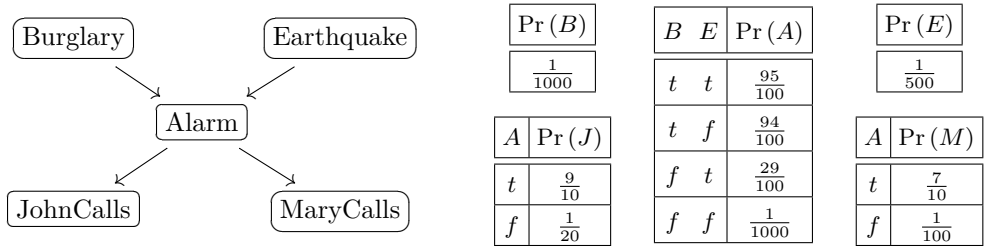
{\xymatrix@R-1pc@C-2pc{
\ovalbox{Burglary}\ar[dr] & & \ovalbox{Earthquake}\ar[dl] 
\\
& \ovalbox{Alarm}\ar[dl]\ar[dr] & 
\\
\ovalbox{JohnCalls} & & \ovalbox{MaryCalls}
}}
\qquad
\begin{tabular}{c}
{\setlength\tabcolsep{0.3em}\begin{tabular}{|c|}
\hline
$\Prob{B}$ \\
\hline\hline
$\frac{1}{1000}$ \\
\hline
\end{tabular}}
\\[2em]
{\setlength\tabcolsep{0.3em}\begin{tabular}{|c|c|}
\hline
$A$ & $\Prob{J}$ \\
\hline\hline
$t$ & $\frac{9}{10}$ \\
\hline
$f$ & $\frac{1}{20}$ \\
\hline
\end{tabular}}
\end{tabular}
\quad
{\setlength\tabcolsep{0.3em}\begin{tabular}{|cc|c|}
\hline
$B$ & $E$ & $\Prob{A}$ \\
\hline\hline
$t$ & $t$ & $\frac{95}{100}$ \\
\hline
$t$ & $f$ & $\frac{94}{100}$ \\
\hline
$f$ & $t$ & $\frac{29}{100}$ \\
\hline
$f$ & $f$ & $\frac{1}{1000}$ \\
\hline
\end{tabular}}
\quad
\begin{tabular}{c}
{\setlength\tabcolsep{0.3em}\begin{tabular}{|c|}
\hline
$\Prob{E}$ \\
\hline\hline
$\frac{1}{500}$ \\
\hline
\end{tabular}}
\\[2em]
{\setlength\tabcolsep{0.3em}\begin{tabular}{|c|c|}
\hline
$A$ & $\Prob{M}$ \\
\hline\hline
$t$ & $\frac{7}{10}$ \\
\hline
$f$ & $\frac{1}{100}$ \\
\hline
\end{tabular}}
\end{tabular}$$

\noindent The probability of each event given its preconditions is as
given in the tables --- for example, the probability that the alarm
rings given that there is a burglar but no earthquake is 0.94.

We model the above question in $\COMET$ as follows.  
$$\begin{array}{l}
\slet{b}{0.01}{\slet{e}{0.002}{}} \\
\qquad\lett a(x,y) = (\ifn x {\begin{array}[t]{l}
   \thenn (\cond{y}{0.95}{0.94}) \\
   \elsen (\cond{y}{0.29}{0.001})) \inn
   \end{array}} \\
\qquad\qquad \slet{j(z)}{(\cond{z}{0.9}{0.05})}{} \\
\qquad\qquad\qquad \slet{m(z)}{(\cond{z}{0.7}{0.01})}{} \\
\pi_{1}\big(\condn{b\sotimes e}{j \after a}\big)
\end{array}$$

\noindent We first elaborate the predicate $j\after a$, given in
context as $x\colon \mathbf{2}, y\colon \mathbf{2} \vdash j(a(x,y))
\colon \mathbf{2}$. It is:
$$\begin{array}{rcl}
j(a(x,y)) 
& = &
\cond{a(x,y)}{0.90}{0.05} \\
& = &
\ifn x {\begin{array}[t]{l} 
  \thenn (\cond{y}{(\cond{0.95}{0.90}{0.05})}{(\cond{0.94}{0.90}{0.05})} \\
  \elsen (\cond{y}{(\cond{0.29}{0.90}{0.05})}{(\cond{0.001}{0.90}{0.05})} 
\end{array}}
\\
& = &
\ifn x {\begin{array}[t]{l} 
  \thenn (\cond{y}{0.95 \andthen 0.90 \ovee 0.95^{\bot} \andthen 0.05}
                  {0.94 \andthen 0.90 \ovee 0.94^{\bot} \andthen 0.05}) \\
  \elsen (\cond{y}{0.29\andthen 0.90 \ovee 0.29^{\bot} \andthen 0.05}
                  {0.001 \andthen 0.90 \ovee 0.001^{\bot} \andthen 0.05} 
\end{array}}
\\
& = &
\ifn x \thenn (\cond{y}{0.8575}{0.849}) \elsen (\cond{y}{0.2965}{0.05085})
\end{array}$$

\noindent The associated assert map is:
$$\begin{array}{rcl}
\assert_{j \after a}(b,e)
& = &
\meas\ {\begin{array}[t]{lcl}
0.001 \andthen 0.002 \andthen 0.8575 
& \mapsto &
\return{\top \sotimes \top} \\
0.001 \andthen 0.998 \andthen 0.849 
& \mapsto &
\return{\top \sotimes \bot} \\
0.999 \andthen 0.002 \andthen 0.2965 
& \mapsto &
\return{\bot \sotimes \top} \\
0.999 \andthen 0.998 \andthen 0.05085
& \mapsto &
\return{\bot \sotimes \bot} \\
0.052138976^\bot
& \mapsto &
\fail 
\end{array}} \\
& = &
\meas\ {\begin{array}[t]{lcl}
0.000001715 
& \mapsto &
\return{\top \sotimes \top} \\
0.000847302 
& \mapsto &
\return{\top \sotimes \bot} \\
0.000592407 
& \mapsto &
\return{\bot \sotimes \top} \\
0.050697552 
& \mapsto &
\return{\bot \sotimes \bot} \\
0.052138976^\bot 
& \mapsto &
\fail 
\end{array}}
\end{array}$$

\noindent Hence by Corollary~\ref{cor:normmeasure} we obtain the
marginalised conditional:
$$\begin{array}{rcl}
\pi_{1}\big(\condn{b\sotimes e}{j \after a}\big)
& = &
\pi_{1}\big(\norm{\assert_{j \after a}(b,e)}\big) \\
& = &
\pi_{1}\big(\meas\ {\begin{array}[t]{lcl}
\nicefrac{0.000001715}{0.052138976}
& \mapsto &
\top \sotimes \top \\
\nicefrac{0.000847302}{0.052138976}
& \mapsto &
\top \sotimes \bot \\
\nicefrac{0.000592407}{0.052138976}
& \mapsto &
\bot \sotimes \top \\
\nicefrac{0.050697552}{0.052138976}
& \mapsto &
\bot \sotimes \bot\,\big) \\
\end{array}} \\
& = &
\meas\ {\begin{array}[t]{lcl}
0.000032893
& \mapsto &
\pi_{1}(\top \sotimes \top) \\
0.016250837
& \mapsto &
\pi_{1}(\top \sotimes \bot) \\
0.011362078
& \mapsto &
\pi_{1}(\bot \sotimes \top) \\
0.972354194
& \mapsto &
\pi_{1}(\bot \sotimes \bot) \\
\end{array}}
\\
& = &
\meas\ {\begin{array}[t]{lcl}
0.000032893
& \mapsto &
\top \\
0.016250837
& \mapsto &
\top \\
0.011362076
& \mapsto &
\bot \\
0.972354194
& \mapsto &
\bot \\
\end{array}}
\\
& = &
\meas\ {\begin{array}[t]{lcl}
0.01628373
& \mapsto &
\top \\
0.98371627
& \mapsto &
\bot \\
\end{array}}
\\
& = &
0.01628373 
\end{array}$$

\noindent We conclude that there is a 1.6\% chance of a burglary when
John calls.

\end{example}

\section{Metatheorems}
\label{section:metatheorems}

We presented an overview of the system in Section \ref{section:overview}, and gave the intuitive meaning of the terms of $\COMET$.
In this section, we proceed to a more formal development of the theory, and investigate what can be proved within the system.

The type theory we have presented enjoys the following standard properties.

\begin{lemma}
\label{lm:meta}
$ $
\begin{enumerate}
\item \textbf{Weakening}
\label{lm:weak}
  If $\Gamma \vdash \mathcal{J}$ and $\Gamma \subseteq \Delta$ then $\Delta \vdash \mathcal{J}$.
\item \textbf{Substitution}
  If $\Gamma \vdash t : A$ and $\Delta, x : A \vdash \mathcal{J}$ then $\Gamma, \Delta \vdash \mathcal{J}[x:=t]$.
\item \textbf{Equation Validity}
  If $\Gamma \vdash s = t : A$ then $\Gamma \vdash s : A$ and $\Gamma \vdash t : A$.
\item \textbf{Inequality Validity}
If $\Gamma \vdash s \leq t : A + 1$ then $\Gamma \vdash s : A + 1$ and $\Gamma \vdash t : A + 1$.
\item \textbf{Functionality}
If $\Gamma \vdash r = s : A$ and $\Delta, x : A \vdash t : B$ then $\Gamma, \Delta \vdash t[x:=r] = t[x:=s] : B$.
\end{enumerate}
\end{lemma}

\begin{proof}
The proof in each case is by induction on derivations.  Each case is straightforward.
\end{proof}

The following lemma shows that substituting within our binding operations works as desired.

\begin{lemma}
\label{lm:sub}
  \begin{enumerate}
  \item \label{lm:letsub}If $\Gamma \vdash r : A \otimes B$; $\Delta, x : A, y : B \vdash s : C$; and $\Theta, z : C \vdash t : D$
then $\Gamma, \Delta, \Theta \vdash t[z:=\plet{x}{y}{r}{s}] = \plet{x}{y}{r}{t[z:=s]} : D$.
\item \label{lm:casesub} If $\Gamma \vdash r : A + B$; $\Delta, x  :A \vdash s : C$; $\Delta, y : B \vdash s' : C$; and $\Theta, z : C \vdash t : D$ then
$$\Gamma, \Delta, \Theta \vdash \begin{array}[t]{l}
t[z:=\pcase{r}{x}{s}{y}{s'}] \\
= \pcase{r}{x}{t[z:=s]}{y}{t[z:=s']} : D
\end{array} \enspace . $$
  \end{enumerate}
\end{lemma}

\begin{proof}
  For part 1, we us the following `trick' to simulate local definition (see \cite{Adams2014}):
\begin{align*}
\lefteqn{t[z := \pcase{r}{x}{s}{y}{s'}]} \\
& = \plet{z}{\_}{(\pcase{r}{x}{s}{y}{s'}) \sotimes *}{t} & \Rbeta \\
& = \plet{z}{\_}{\pcase{r}{x}{s \sotimes *}{y}{s' \sotimes *}}{t} & \Rcasepair \\
& = \pcase{r}{x}{\plet{z}{\_}{s \sotimes *}{t}}{y}{\plet{z}{\_}{s' \sotimes *}{t}} & \Rletcase \\
& = \pcase{r}{x}{t[z:=s]}{y}{t[z:=s']} & \Rbeta
\end{align*}
Part 2 is proven similarly using\Rletpair and\Rletlet.
\end{proof}

\begin{corollary}
\label{cor:vacsub}
  \begin{enumerate}
  \item If $\Gamma \vdash s : A \otimes B$ and $\Delta \vdash t : C$ then
$\Gamma, \Delta \vdash \plet{\_}{\_}{s}{t} = t : C$.
  \item \label{cor:vaccase} If $\Gamma \vdash s : A + B$ and $\Delta \vdash t : C$ then $\Gamma, \Delta \vdash \pcase{s}{\_}{t}{\_}{t} = t : C$.
  \end{enumerate}
\end{corollary}

\begin{proof}
  These are both the special case where $z$ does not occur free in $t$.
\end{proof}

\subsection{Coproducts}

We generalise the $\inlprop{}$ and $\inrprop{}$ constructions as follows.  
Define the predicate $\intest{i}{}$ on $n \cdot A$, which tests whether a term
comes from the $i$th component, as follows.
\[ \intest{i}{t} \eqdef \case_{j=1}^n t \of \nin{j}{n}{\_} \mapsto \begin{cases}
\top & \text{if } i = j \\
\bot & \text{if } i \neq j
\end{cases} \]




\subsection{The Do Notation}

Our construction $\doo{x}{s}{t}$ satisfies the following laws.

\begin{lemma}
\label{lm:do}
Let $\Gamma \vdash r : A + 1$, $\Delta, x : A \vdash s : B + 1$, and $\Theta, y : B \vdash t : C$.  Let
also $\Gamma \vdash r' : A$.  Then
\begin{align*}
\Gamma, \Delta \vdash \doo{x}{\return{r'}}{s} & = t[x:=s] : B + 1 \\
\Gamma, \Delta \vdash \doo{x}{\fail}{s} & = \fail : B + 1 \\
\Gamma \vdash \doo{x}{r}{\return{x}} & = r : A + 1 \\
\Gamma \vdash \doo{\_}{r}{\fail} & = \fail : B + 1 \\
\Gamma, \Delta, \Theta \vdash \doo{x}{r}{(\doo{y}{s}{t})} & = \doo{y}{(\doo{x}{r}{s})}{t}  : C
\end{align*}
\end{lemma}

\begin{proof}
  These all follow easily from the rules for coproducts\Rbetaplusone,\Rbetaplustwo,\Retaplus and\Rcasecase.
\end{proof}

\subsection{Kernels}

\begin{lemma}$ $
\label{lm:kernel}
\begin{enumerate}
\item 
If $\Gamma \vdash t : A + 1$ then $\Gamma \vdash \dom{t} = (\doo{\_}{t}{\top}) : \mathbf{2}$
\item 
\label{lm:kernel2}
Let $\Gamma \vdash t : A + 1$.  Then
$\Gamma \vdash \dom{t} = \bot : \mathbf{2}$ if and only if $\Gamma \vdash t = \fail : A + 1$.
\item 
Let $\Gamma \vdash s : A + 1$ and $\Delta, x : A \vdash t : B + 1$.  Then
$\Gamma, \Delta \vdash \dom{(\doo{x}{s}{t})} = \doo{x}{s}{\dom{t}} : \mathbf{2}$.
\end{enumerate}
\end{lemma}

\begin{proof}$ $
\begin{enumerate}
\item 
This holds just by expanding definitions.
\item
Obviously, $(\dom{\fail}) = \bot$.  For the converse, if $\dom{t} =
\bot$ then $\ker{t} = \top$ and so $t = \inr{\rgt{t}} = \inr{*}$ by\Retaone.
\item
$ \begin{aligned}[t]
(\dom{\pcase{s}{x}{t}{\_}{\fail}}) & = \pcase{s}{x}{\dom{t}}{\_}{\dom{\fail}} \\
& = \pcase{s}{x}{\dom{t}}{\_}{\bot} 
\end{aligned} $
\end{enumerate}
\end{proof}

\subsection{Finite Types}

\begin{lemma}
\label{lm:rhdfin}
Let $\Gamma \vdash t : \mathbf{n}$ and $i \leq n$.
If $\Gamma \vdash \rhd_i(t) = \top : \mathbf{2}$ then $\Gamma \vdash t = i : \mathbf{n}$.
\end{lemma}

\begin{proof}
Define $x : \mathbf{n} \vdash f(x) : 1 + \mathbf{n-1}$ by
\[ f(x) \eqdef \case_{j=1}^n x \of \begin{cases}
\inr{j} & \text{if } j < i \\
\inl{*} & \text{if } j = i \\
\inr{j-i} & \text{if } j > i
\end{cases} \]
Then $\Gamma \vdash \inlprop{f(t)} = \top : \mathbf{2}$, hence
\[ f(t) = \inl{\lft{f(t)}} = \inl{*} \]
We can define an inverse to $f$: given $x : 1 + \mathbf{n-1}$, define
\[ f^{-1}(x) \eqdef \case x \of \inl{\_} \mapsto i \mid \inr{t} \mapsto \case_{j=1}^{n-1} t \of j \text{ if } j < i \mid j + 1 \text{ if } j \geq i \]
Then $x : \mathbf{n} \vdash f^{-1}(f(x)) = x : 1 + \mathbf{n-1}$ and so $\Gamma \vdash t = f^{-1}(f(t)) = f^{-1}(\inl{*}) = i : \mathbf{n}$.
\end{proof}

\subsection{Ordering on Partial Maps and the Partial Sum}
\label{section:psum}
Note that, from the rules\Roveeprime and\Roveedef, we have $\Gamma \vdash s \ovee t : A + 1$ if and only if there exists $\Gamma \vdash b : (A + A) + 1$ such that
\[ \Gamma \vdash b \goesto \rhd_1 = s : A + 1, \qquad \Gamma \vdash b \goesto \rhd_2 = t : A + 1 \enspace , \]
in which case $\Gamma \vdash s \ovee t = \doo{x}{b}{\return{\nabla(x)}} : A + 1$.  We say that such a term $b$ is a \emph{bound} for $s \ovee t$.  By
the rule\RJMprime, this bound is unique if it exists.

\begin{lemma}
For predicates $\Gamma \vdash p, q : \mathbf{2}$, we have that $\Gamma \vdash b : \mathbf{3}$ is a bound for $p \ovee q$ if and only if $\rhd_1(b) = p$ and $\rhd_2(b) = q$.
\end{lemma}

\begin{proof}
  This holds because $b \goesto \rhd_1 = \rhd_1(b)$ and $b \goesto \rhd_2 = \rhd_2(b)$, as can be seen just from expanding definitions.
\end{proof}

The set of \emph{partial} maps $A \rightarrow B + 1$ between any two types $A$ and $B$ form a \emph{partial commutative monoid} (PCM) with least element $\fail$, as shown by the following results.

\begin{lemma}$ $
  \label{lm:ordering}
  \begin{enumerate}
  \item \label{lm:zerolaw} If $\Gamma \vdash t : A + 1$ then $\Gamma \vdash t \ovee \fail = t : A + 1$.
  \item (\textbf{Commutativity}) If $\Gamma \vdash s \ovee t : A + 1$ then $\Gamma \vdash t \ovee s : A + 1$ and $\Gamma \vdash s \ovee t = t \ovee s : A + 1$.
  \item (\textbf{Associativity}) $\Gamma \vdash (r \ovee s) \ovee t : A + 1$ if and only if $\Gamma \vdash r \ovee (s \ovee t) : A + 1$, in which case $\Gamma \vdash r \ovee (s \ovee t) = (r \ovee s) \ovee t : A + 1$. \label{lm:assoc}
  \end{enumerate}
\end{lemma}

\begin{proof}
  \begin{enumerate}
  \item The bound is $\doo{x}{t}{\return{\inl{x}}}$.
\item 
Let $b$ be a bound for $s \ovee t$.  Then $\doo{x}{b}{\return{\swapper{x}}}$ is a bound for $t \ovee s$ and we have
\begin{align*}
t \ovee s & = \doo{y}{(\doo{x}{b}{\return{\swapper{x}}})}{\return{\nabla(y)}} \\
& = \doo{x}{b}{\doo{y}{\return{\swapper{x}}}{\return{\nabla(y)}}} \\
& = \doo{x}{b}{\return{\nabla(\swapper{x})}} = \doo{x}{b}{\return{\nabla(x)}} \\
& = s \ovee t
\end{align*}
\item 
This is proved in Appendix \ref{section:associativity}
\end{enumerate}
\end{proof}

\begin{lemma}
\label{lm:oveeleq}
Let $\Gamma \vdash r : A + 1$ and $\Gamma \vdash s : A + 1$.  Then $\Gamma \vdash r \leq s : A + 1$ if and only if there exists $t$ such that $\Gamma \vdash r \ovee t = s : A + 1$.
\end{lemma}

\begin{proof}
Suppose $r \leq s$.  If $b$ is such that $\doo{x}{b}{\rhd_1(x)} = r$ and $\doo{x}{b}{\return{\nabla(x)}} = s$ then take $t = \doo{x}{b}{\rhd_2(x)}$.

Conversely, if $r \ovee t = s$, then inverting the derivation of $\Gamma \vdash r \ovee t : A + 1$ we have that there exists $b$ such that $r =
\doo{x}{b}{\rhd_1(x)}$, $t = \doo{x}{b}{\rhd_2(x)}$ and $s = r \ovee t = \doo{x}{b}{\return{\nabla(x)}}$.  Therefore, $r \leq s$ by\RleqI.
\end{proof}

\begin{corollary}
  Let $\Gamma \vdash r : A + 1$ and $\Gamma \vdash s : A + 1$.  Then $\Gamma \vdash r \leq s : A + 1$ if and only if there exists $b$ such that $\Gamma \vdash b : (A + A) + 1$,
$\Gamma \vdash b \goesto \rhd_1 = s : A + 1$, and $\Gamma \vdash \doo{x}{b}{\return{\nabla(x)}} = s : A + 1$.
\end{corollary}

This term $b$ is called a \emph{bound} for $s \leq t$.

Using this characterisation of the ordering relation, we can read off several properties directly from Lemma \ref{lm:ordering}.

\begin{lemma}
\begin{enumerate}
  \item If $\Gamma \vdash s \ovee t : A + 1$ then $\Gamma \vdash s \leq s \ovee t : A + 1$ and $\Gamma \vdash t \leq s \ovee t : A + 1$. \label{lm:leqovee}
  \item If $\Gamma \vdash t : A + 1$ then $\Gamma \vdash t \leq t : A + 1$.
  \item If $\Gamma \vdash t : A + 1$ then $\Gamma \vdash \fail \leq t : A + 1$.
  \item If $\Gamma \vdash r \leq s : A + 1$ and $\Gamma \vdash s \leq t : A + 1$ then $\Gamma \vdash r \leq t : A + 1$.
  \item If $\Gamma \vdash r \leq s : A + 1$ and $\Gamma \vdash s \ovee t : A + 1$ then $\Gamma \vdash r \ovee t \leq s \ovee t : A + 1$.
\end{enumerate}
\end{lemma}

\begin{proof}
\begin{enumerate}
\item 
From Lemma \ref{lm:oveeleq} and Commutativity.
\item From Lemma \ref{lm:oveeleq} and Lemma \ref{lm:ordering}.\ref{lm:zerolaw}.
\item From Lemma \ref{lm:oveeleq} and Lemma \ref{lm:ordering}.\ref{lm:zerolaw}.
\item From Lemma \ref{lm:oveeleq} and Associativity.
\item 
Let $r \ovee x = s$.  Then $r \ovee x \ovee t = s \ovee t$ and so $r \ovee t \leq s \ovee t$.
\end{enumerate}
\end{proof}

On the predicates, we have the following structure, which shows that they form an \emph{effect algebra}.  (In fact, they have more structure:
they form an \emph{effect module} over the scalars, as we will prove in Proposition \ref{prop:effmod}.)

\begin{proposition}
\label{prop:logic}
Let $\Gamma \vdash p,q,r : \mathbf{2}$.
  \begin{enumerate}
  \item If $\Gamma \vdash p : \mathbf{2}$ then $\Gamma \vdash p \ovee p^\bot = \top : \mathbf{2}$.
  \item If $\Gamma \vdash p \ovee q = \top : \mathbf{2}$ then $\Gamma \vdash q = p^\bot : \mathbf{2}$.
  \item (\textbf{Zero-One Law}) If $\Gamma \vdash p \ovee \top : \mathbf{2}$ then $\Gamma \vdash p = \bot : \mathbf{2}$.
  \item \label{prop:ortho} $\Gamma \vdash p \ovee q : \mathbf{2}$ if and only if $\Gamma \vdash p \leq q^\bot : \mathbf{2}$.
  \end{enumerate}
\end{proposition}

\begin{proof}
\begin{enumerate}
\item 
The term $\inl{p} : \mathbf{2} + 1$ is a bound for $p \ovee p^\bot$, and $\doo{x}{\inl{p}}{\return{\nabla(x)}} = \top$.
\item 
Let $b$ be a bound for $p \ovee q$.  We have
\begin{align*}
\top & = \doo{x}{b}{\return{\nabla(x)}} = \doo{x}{b}{\top} & \text{using \Retaone} \\
& = \dom{b}
\end{align*}
Therefore, $b = \inl{\lft{b}}$ by\Rbetaleft, and so
\begin{align*}
  p & = \rhd_1(\lft{b}), \qquad q = \rhd_2(\lft{b}) = \rhd_1(\lft{b})^\bot = p^\bot
\end{align*}
\item Let $b$ be a bound for $p \ovee \top$.  Then $\rhd_2(b) = \top$ and so $b = 2 : \mathbf{3}$ by Lemma \ref{lm:rhdfin}.  Therefore, $p = \rhd_1(b) = \bot$.
\item
Suppose $p \ovee q : \mathbf{2}$.  Then $p \ovee q \ovee (p \ovee q)^\bot = \top$, hence $p \ovee (p \ovee q)^\bot = q^\bot$, and
thus $p \leq q^\bot$.

Conversely, if $p \leq q^\bot$, let $p \ovee x = q^\bot$.  Then $\top = q \ovee q^\bot = p \ovee q \ovee x$, and so $p \ovee q : \mathbf{2}$.
\end{enumerate}
\end{proof}

\begin{corollary}
  \begin{enumerate}
  \item (\textbf{Cancellation}) If $\Gamma \vdash p \ovee q = p \ovee r : \mathbf{2}$ then $\Gamma \vdash q = r : \mathbf{2}$.
  \item (\textbf{Positivity}) If $\Gamma \vdash p \ovee q = \bot : \mathbf{2}$ then $\Gamma \vdash p = \bot : \mathbf{2}$ and $\Gamma \vdash q = \bot : \mathbf{2}$.
  \item If $\Gamma \vdash p : \mathbf{2}$ then $\Gamma \vdash p \leq \top : \mathbf{2}$.
  \item If $\Gamma \vdash p \leq q : \mathbf{2}$ then $\Gamma \vdash q^\bot \leq p^\bot : \mathbf{2}$.
  \end{enumerate}
\end{corollary}

\begin{proof}
  \begin{enumerate}
  \item We have
    \begin{gather*}
      p \ovee q \ovee (p \ovee q)^\bot = p \ovee r \ovee (p \ovee q)^\bot = \top \\
\therefore q = r = (p \ovee (p \ovee q)^\bot)^\bot
    \end{gather*}
  \item 
If $p \ovee q = \bot$ then $p \ovee q \ovee \top : \mathbf{2}$, hence $p \ovee \top : \mathbf{2}$ by Associativity, and so $p = \bot$ by the Zero-One Law.
  \item
We have $p \ovee p^\bot = \top$ and so $p \leq \top$.
\item 
Let $p \ovee x = q$.  Then $\top = q \ovee q^\bot = p \ovee x \ovee q^\bot$, and so $p^\bot = x \ovee q^\bot$.  Thus, $q^\bot \leq p^\bot$.
  \end{enumerate}
\end{proof}

Our next lemma shows how $\ovee$ and $\mathsf{case}$ interact.

\begin{lemma}
\label{lm:caseovee}
Suppose $\Gamma \vdash r : A + B$ and $\Delta, x : A \vdash s \ovee t : C + 1$ and $\Delta, y : B \vdash s' \ovee t' : C + 1$.  Then
\[ \Gamma, \Delta \vdash \begin{array}[t]{l}
\pcase{r}{x}{s \ovee t}{y}{s' \ovee t'} \\
= (\pcase{r}{x}{s}{y}{s'}) \ovee (\pcase{r}{x}{t}{y}{t'}) : C + 1
\end{array}  \]
\end{lemma}

\begin{proof}
  Let $b(x)$ be a bound for $s \ovee t$ in $\Delta, x : A$, and $c(y)$ a bound for $s' \ovee t'$ in $\Delta, y : B$.  Then
\[ \pcase{r}{x}{b(x)}{y}{c(y)} : (B + B) + 1 \]
is a bound for $(\pcase{r}{x}{s}{y}{s'}) \ovee (\pcase{r}{x}{t}{y}{t'})$, and so
\begin{align*}
\lefteqn{(\pcase{r}{x}{s}{y}{s'}) \ovee (\pcase{r}{x}{t}{y}{t'})} \\
& = \doo{z}{\pcase{r}{x}{b(x)}{y}{c(y)}}{\return{\nabla(z)}} \\
& = \pcase{r}{x}{\doo{z}{b(x)}{\return{\nabla(z)}}}{y}{\doo{z}{c(y)}{\return{\nabla(z)}}} \\
& = \pcase{r}{x}{s \ovee t}{y}{s' \ovee t'}
\end{align*}
\end{proof}

\begin{corollary}
\label{cor:doovee}
  If $\Gamma \vdash r : A + 1$ and $\Delta, x : A \vdash s \ovee t : B + 1$ then
\[ \Gamma, \Delta \vdash \doo{x}{r}{s \ovee t} = (\doo{x}{r}{s}) \ovee (\doo{x}{r}{t}) : B + 1 \enspace . \]
\end{corollary}

\begin{proof}
  \begin{align*}
    \doo{x}{r}{s \ovee t} = & \pcase{r}{x}{s \ovee t}{\_}{\fail \ovee \fail} \\
= & (\pcase{r}{x}{s}{\_}{\fail}) \ovee \\
& (\pcase{r}{x}{t}{\_}{\fail})
  \end{align*}
\end{proof}

The following lemma relates the structures on partial maps and predicates via the domain operator.

\begin{lemma}
  If $\Gamma \vdash s \ovee t : A + 1$ then $\Gamma \vdash \dom{(s \ovee t)} = \dom{s} \ovee \dom{t} : \mathbf{2}$.
\end{lemma}

\begin{proof}
  Let $b$ be a bound for $s \ovee t$.  Then
\[ \dom{(s \ovee t)} = \dom{(\doo{x}{b}{\return{\nabla(x)}})} = \doo{x}{b}{\top} = \dom{b} \]
We also have
\begin{align*}
  \dom{s} & = \doo{x}{b}{\inlprop{x}}, \qquad \dom{t} = \doo{x}{b}{\inrprop{x}} \\
\therefore \dom{s} \ovee \dom{t} & = \doo{x}{b}{\inlprop{x} \ovee \inrprop{x}} & (\text{previous part}) \\
& = (\doo{x}{b}{\top}) = \dom{b}
\end{align*}
\end{proof}

Using this, we can conclude several properties about partial maps immediately from the fact that they hold for predicates:

\begin{lemma}
\begin{enumerate}
  \item (\textbf{Restricted Cancellation Law}) If $\Gamma \vdash s \ovee t = t : A + 1$ then $\Gamma \vdash s = \fail : A + 1$.
  \item (\textbf{Positivity}) If $\Gamma \vdash s \ovee t = \fail : A + 1$ then $\Gamma \vdash s = \fail : A + 1$ and $\Gamma \vdash t = \fail : A + 1$.
  \item If $\Gamma \vdash s \leq t : A + 1$ and $\Gamma \vdash t \leq s : A + 1$ then $\Gamma \vdash s = t : A + 1$.
\end{enumerate}
\end{lemma}

\begin{proof}
\begin{enumerate}
\item 
Suppose $\Gamma \vdash s \ovee t = t : A + 1$.  Then $\Gamma \vdash \dom{(s \ovee t)} = \dom{s} \ovee \dom{t} = \dom{t} : \mathbf{2}$,
and so $\Gamma \vdash \dom{s} = \bot : \mathbf{2}$ and $\Gamma \vdash s = \fail : A + 1$ by Lemma \ref{lm:kernel}.\ref{lm:kernel2}.
\item Suppose $\Gamma \vdash s \ovee t = \fail$.  Then $\dom{(s \ovee
  t)} = \dom{s} \ovee \dom{t} = \bot$, and so $\dom{s} = \bot$ and
$\dom{t} = \bot$.  Therefore, $s = \fail$ and $t = \fail$ by Lemma
\ref{lm:kernel}.\ref{lm:kernel2}.
\item 
Let $s \ovee b = t$ and $t \ovee c = s$.  Then $s \ovee b \ovee c = s$ and so $b \ovee c = \fail$ by the Restricted Cancellation Law, hence $b = c = \fail$ by Positivity.
Thus, $s = s \ovee \fail = t$.
  \end{enumerate}
\end{proof}

Finally, we can show that the partial projections on copowers behave as expected with respect to $\ovee$.

\begin{lemma}
For $t : n \cdot A$,
\[ \rhd_{i_1, \ldots, i_k}(t) = \rhd_{i_1}(t) \ovee \cdots \ovee \rhd_{i_k}(t) \]
\end{lemma}

\begin{proof}
The proof is by induction on $k$.  Take \[ b = \case_{i=1}^n t \of \nin{i}{n}{\_} \mapsto \begin{cases} 1 & \text{if } i = i_1, \ldots, i_k\\
2 & \text{if } i = i_{k+1} \\
3 & \text{otherwise}
\end{cases}\]
Then $\rhd_1(b) = \rhd_{i_1\cdots i_k}(t)$, $\rhd_2(b) = \rhd_{i_{k+1}}(t)$, and $\rhd_{12}(b) = \rhd_{i_1\cdots i_ki_{k+1}}(t)$.  Therefore,
\[ \rhd_{i_1i_2\cdots i_{k+1}}(t) = \rhd_{i_1\cdots i_k}(t) \ovee \rhd_{i_{k+1}}(t) = \rhd_{i_1}(t) \ovee \cdots \ovee \rhd_{i_{k+1}}(t) \]
by the induction hypothesis.
\end{proof}

\subsubsection{Assert Maps}
\label{section:assert}

Recall that, for $x : A \vdash p : \mathbf{2}$ and $\Gamma \vdash t : A$, we define $\Gamma \vdash \assert_{\lambda x p}(t) \eqdef \rhd_1(\instr_{\lambda x p}(t)) : A + 1$.

This operation $\assert$ forms a bijection between:
\begin{itemize}
\item
the terms $p$ such that $x : A \vdash p : \mathbf{2}$ (the predicates on $A$); and
\item 
the terms $t$ such that $x : A \vdash t \leq \return{x} : A + 1$
\end{itemize}

This is proven by the following result.

\begin{lemma}
\label{lm:assert}
If $x : A \vdash p : 1 + 1$ and $\Gamma \vdash t : A$, then
\begin{enumerate}
\item 
$\Gamma \vdash \assert_{\lambda x p}(t) : A + 1$
\item 
$\Gamma \vdash \assert_{\lambda x p}(t) \leq \inl{t} : A + 1$.
\item \textbf{\Rassertdown}
\label{lm:assertdown}
$\Gamma \vdash \dom{\assert_{\lambda x p}(t)} = [t/x] p : \mathbf{2}$
\item
If $x : A \vdash t \leq \inl{x} : A + 1$ then $x : A \vdash t = \assert_{\lambda x (\dom{t})}(x) : A + 1$.
\end{enumerate}
\end{lemma}

\begin{proof}
\begin{enumerate}
\item 
An easy application of the rules\Rinstr,\Rcase,\Rinl,\Rinr and\Runit.
\item 
The term $\inl{\instr_{\lambda x p}(t)}$ is a bound for this inequality.
\item 
\begin{align*}
\dom{\assert_{\lambda x p}(t)} & \eqdef \dom{\rhd_1(\instr_{\lambda x p}(t))} = \inlprop{\instr_{\lambda x p}(t)} \\
& = p[x:=t] & \text{by\Rinstrtest}
\end{align*}
\item
Let $b$ be a bound for the inequality $t \leq \inl{x}$, so $(b \goesto \rhd_1) = t$ and $\doo{x}{b}{\return{\nabla(x)}} = \inl{x}$.
Then
\[ \dom{b} = \dom{(\doo{x}{b}{\return{\nabla(x)}})} = \dom{\inl{x}} = \top . \]
Hence we can define $c = \lft{b}$.  We therefore have $\rhd_1(c) = t$ and $\nabla(c) = x$.
Now, the rule\Retainstr gives us
\begin{gather*}
c = \instr_{\lambda x \inlprop{c}}(x) = \instr_{\dom{\lambda x t}}(x) \\
\therefore t = \rhd_1(c) = \assert_{\dom{\lambda x t}}(x)
\end{gather*}
\end{enumerate}
\end{proof}

We now give rules for calculating $\instr_{\lambda x p}$ and $\assert_{\lambda x p}$ directed by the type.

\begin{lemma}[\Rassertscalar]
If $\vdash s : \mathbf{2}$ then
\[ \vdash \assert_{\lambda \_ s}(*) = \instr_{\lambda \_ s}(*) = s : \mathbf{2} \]
\end{lemma}

\begin{proof}
We have $\nabla(s) = *$ by\Retaone and $\dom{s} = s$ by\Retaplus.  The
result follows by\Retainstr.
\end{proof}

\begin{lemma}[\Rinstrplus,\Rassertplus]
  If $x : A + B \vdash p : \mathbf{2}$ and $\Gamma \vdash t : A + B$ then
\begin{align*}
\Gamma \vdash \instr_{\lambda x p}(t) = \case t \of
& \inl{y} \mapsto (\inln + \inln)(\instr_{\lambda a. p[x:=\inl{a}]}(y)) \mid \\
& \inr{z} \mapsto (\inrn + \inrn)(\instr_{\lambda b. p[x:=\inr{b}]}(z)) \\
\Gamma \vdash \assert_{\lambda x p}(t) = \case t \of
& \inl{y} \mapsto \doo{w}{\assert_{\lambda a. p[x:=\inl{a}]}(y)}{\return{\inl{w}}} \mid \\
& \inr{z} \mapsto \doo{w}{\assert_{\lambda b.p[x:=\inr{b}]}(z)}{\return{\inr{w}}}
\end{align*}
where $(\inln + \inln)(t) \eqdef \pcase{t}{x}{\inl{x}}{y}{\inl{y}}$, and $(\inrn + \inrn)(t)$ is defined similarly.
\end{lemma}

\begin{proof}
For $x : A + B$, let us write $f(x)$ for
\begin{align*}
f(x) \eqdef
\case x \of
& \inl{y} \mapsto (\inln + \inln)(\instr_{\lambda a.p[\inl{a}]}(y)) \mid \\
& \inr{z} \mapsto (\inrn + \inrn)(\instr_{\lambda b.p[\inr{b}]}(z))
\end{align*}
We shall prove $f(x) = \instr_{\lambda x p}(x)$.

We have 
\begin{align*}
\nabla(f(x)) & = \case x \of \begin{array}[t]{l} \inl{y} \mapsto
\inl{\nabla(\assert_{\lambda a.p[x:=\inl{a}]}(y))} \mid \\
\inr{z} \mapsto \inr{\nabla(\assert_{\lambda b.p[\inr{b}]}(z))} 
\end{array} \\
& = \pcase{x}{y}{\inl{y}}{z}{\inr{z}} \\
& = x & \text{by \Retaplus} \\
  \dom{f(x)} & = \pcase{x}{y}{\dom{\instr_{\lambda a.p[x:=\inl{a}]}(y)}}{z}{\instr_{\lambda b.p[\inr{b}]}(z)} \\
& = \pcase{x}{y}{p[x:=\inl{y}]}{z}{p[x:=\inr{z}]} \\
& = p & \text{by Corollary \ref{cor:vacsub}.\ref{cor:vaccase}}
\end{align*}
Hence $f(x) = \instr_p(x)$ by\Retainstr.
\end{proof}

\begin{corollary}[\Rinstrm,\Rassertm]
\label{cor:assertn}
\begin{enumerate}
\item 
Given $x : \mathbf{m} \vdash t : \mathbf{n}$ and $\Gamma \vdash s : \mathbf{m}$,
\[ \instr_{\lambda x t}(s) = \case_{i=1}^m\ s \of i \mapsto \case_{j=1}^n\ t[x:=i] \of j \mapsto \nin{j}{n}{i} \enspace . \]
\item 
Given $x : \mathbf{n} \vdash p : \mathbf{2}$ and $\Gamma \vdash t : \mathbf{n}$,
\[ \assert_p(t) = \case_{i=1}^n t \of i \mapsto \cond{p[x:=i]}{\return{i}}{\fail} \enspace . \]
\end{enumerate}
\end{corollary}

\subsection{Sequential Product}

We do not have conjunction or disjunction in our language for predicates over the same type, as this would involve duplicating variables.  However, we do have
the following \emph{sequential product}.
(This was called the `and-then' test operator in Section 9 in \cite{Jacobs14}.)

Let $x : A \vdash p,q : \mathbf{2}$.  We define the \emph{sequential product} $p \andthen q$ by
\[ x : A \vdash p \andthen q \eqdef \doo{x}{\assert_{\lambda x p}(x)}{q} : \mathbf{2} \enspace . \]

\begin{proposition}$ $
\label{prop:testops}
\label{prop:effmod}
Let $x : A \vdash p,q : \mathbf{2}$.
\begin{enumerate}
\item $\instr_{p \andthen q}(x) = \pcase{\instr_p(x)}{x}{\instr_q(x)}{y}{\inr{y}}$
\item $\assert_{p \andthen q}(x) = \doo{x}{\assert_p(x)}{\assert_q(x)} \eqdef \assert_p(x) \goesto \assert_q$ \label{prop:assertand}
  \item \label{prop:odotcomm} (\textbf{Commutativity})
$p \andthen q = q \andthen p$.
  \item $(p \ovee q) \andthen r = p \andthen r \ovee q \andthen r$ and $p \andthen (q \ovee r) = p \andthen q \ovee p \andthen r$.
  \item $p \andthen \bot = \bot \andthen q = \bot$
  \item $p \andthen \top = p$ and $\top \andthen q = q$
  \item $p \andthen (q \andthen r) = (p \andthen q) \andthen r$
  \item Let $x : A \vdash p : \mathbf{2}$.  If $x$ does not occur in $q$, then $p \andthen q = \pcase{p}{\_}{q}{\_}{\bot}$.
\end{enumerate}
\end{proposition}

\begin{proof}
\begin{enumerate}
\item
We have
\begin{align*}
&  \inlprop{\pcase{\instr_p(x)}{x}{\instr_q(x)}{y}{\inr{y}}} \\
& = \pcase{\instr_p(x)}{x}{q}{y}{\bot} \\
& = \doo{x}{\assert_p(x)}{q} = p \andthen q
\end{align*}
and
\begin{align*}
& \nabla(\pcase{\instr_p(x)}{x}{\instr_q(x)}{y}{\inr{y}}) \\
& = \pcase{\instr_p(x)}{x}{x}{y}{y} \\
& = \nabla(\instr_p(x)) = x
\end{align*}
so the result follows by\Retainstr.
\item This follows immediately from the previous part.
\item This follows from the previous part and the rule\Rcomm (Appendix \ref{section:instruments}).
\item 
$p \andthen (q \ovee r) = (p \andthen q) \ovee (p \andthen r)$ by Corollary \ref{cor:doovee}.  The other case follows by Commutativity.
\item 
$\bot \andthen p = \bot$ by Lemma \ref{lm:do}.
\item 
$\top \andthen q = q$ by Lemma \ref{lm:do}.
\item $
\begin{aligned}[t]
(p \andthen q) \andthen r
& \eqdef \doo{x}{\assert_{p \andthen q}(x)}{r} \\
& = \doo{x}{(\assert_p(x) \goesto \assert_q)}{r} & \text{by part \ref{prop:assertand}} \\
& = \doo{x}{\assert_p(x)}{\doo{x}{\assert_q(x)}{r}} & \text{by Lemma \ref{lm:do}} \\
& \eqdef p \andthen (q \andthen r)
  \end{aligned} $
\item 
$p \andthen q = \doo{\_}{\assert_p(x)}{q} = \pcase{\assert_p(x)}{\_}{q}{\_}{\bot} = \pcase{\dom{(\assert_p(x))}}{\_}{q}{\_}{\bot} = \cond{p}{q}{\bot}$.
\item 
Let $b : \mathbf{3}$ be given by
\[ b \eqdef \cond{p}{\cond{q}{1}{3}}{\cond{r}{2}{3}} \]
Then
\begin{align*}
  b \goesto \rhd_1 & = \cond{p}{\cond{q}{\top}{\bot}}{\cond{r}{\bot}{\bot}} \\
& = \cond{p}{q}{\bot} = p \andthen q \\
b \goesto \rhd_2 & = \cond{p}{\bot}{r} & \text{similarly} \\
& = \cond{p^\bot}{r}{\bot} = p^\bot \andthen r
\end{align*}
Thus, $b$ is a bound for $p \andthen q \ovee p^\bot \andthen r$.  We also have
\begin{align*}
\doo{x}{b}{\return{\nabla(x)}} & \eqdef \cond{p}{\cond{q}{\top}{\bot}}{\cond{r}{\top}{\bot}} \\
& = \cond{p}{q}{r}
\end{align*}
and the result is proved.
\end{enumerate}
\end{proof}

These results show that the scalars form an \emph{effect monoid}, and the predicates on any type form an \emph{effect module} over that effect monoid (see \cite{Jacobs14} Lemma 13 and Proposition 14).

\subsection{n-tests}
\label{section:ntest}

Recall that an \emph{$n$-test} on a type $A$ is an $n$-tuple $(p_1, \ldots, p_n)$ such that
\[ x : A \vdash p_1 \ovee \cdots \ovee p_n = \top : \mathbf{2} \]

The following lemma shows that there is a one-to-one correspondance between the $n$-tests on $A$, and
the maps $A \rightarrow \mathbf{n}$.

\begin{lemma}
\label{lm:ntest}
  For every $n$-test $(p_1, \ldots, p_n)$ on $A$, there exists a term $x : A \vdash t(x) : \mathbf{n}$,
unique up to equality, such that
\[ x : A \vdash p_i(x) = \rhd_i(t(x)) : \mathbf{2} \]
\end{lemma}

\begin{proof}
The proof is by induction on $n$.  The case $n = 1$ is trivial.

Suppose the result is true for $n$.  Take an $n+1$-test $(p_1, \ldots, p_{n+1})$.  Then \\
$(p_1, p_2, \ldots, p_n \ovee p_{n+1})$ is an $n$-test.  By the induction hypothesis, there exists $t : \mathbf{n}$ such that
\[ \rhd_i(t) = p_i \; (i < n), \qquad \rhd_n(t) = p_n \ovee p_{n+1} \enspace . \]
Let $b : \mathbf{3}$ be the bound for $p_n \ovee p_{n+1}$, so
\[ \rhd_1(b) = p_n, \qquad \rhd_2(b) = p_{n+1}, \qquad \rhd_{12}(b) = p_n \ovee p_{n+1} \enspace . \]
Reading $t$ and $b$ as partial functions in $\mathbf{n-1} + 1$ and $\mathbf{2} + 1$, we have that
$\ker{t} = \dom{b} = p_n \ovee p_{n+1}$.
Hence $\inlr{b}{t} : \mathbf{2} + \mathbf{n - 1}$ exists.  Reading it as a term of type $\mathbf{n+1}$, we have that
\[ \rhd_1(\inlr{b}{t}) = p_n, \quad \rhd_2(\inlr{b}{t}) = p_{n+1}, \quad \rhd_{i + 2}(\inlr{b}{t}) = p_i \; (i < n) \enspace . \]
From this it is easy to construct the term of type $\mathbf{n + 1}$ required.
\end{proof}

We write $\instr_{(p_1, \ldots, p_n)}(s)$ for $\instr_t(s)$, where $t$ is the term such that $\rhd_i(t) = p_i$ for each $i$.
We therefore have

\begin{lemma}
\label{lm:instrn}
  $\instr_{(p_1, \ldots, p_n)}(x)$ is the unique term such that $\intest{i}{\instr_{(p_1, \ldots, p_n)}(x)} = p_i$ for all $i$
and
$\nabla(\instr_{(p_1, \ldots, p_n)}(x)) = x$.
\end{lemma}

\begin{proof}
Let $t : \mathbf{n}$ be the term such that $\rhd_i(t) = p_i$ for all $i$.
By the rules for instruments, $\instr_{(p_1, \ldots, p_n)}(x)$ is the unique term such that
\begin{align*}
  (\case_{i=1}^n \instr_{(p_1, \ldots, p_n)}(x) \of \nin{i}{n}(\_) \mapsto i) & = t \\
\nabla(\instr_{(p_1, \ldots, p_n)}(x)) & = x
\end{align*}
It is therefore sufficient to prove that, given terms $\Gamma \vdash s, t : \mathbf{n}$,
\[ \Gamma \vdash s = t : \mathbf{n} \Leftrightarrow \forall i. \Gamma \vdash \rhd_i(s) = \rhd_i(t) : \mathbf{2} \]
This fact is proven by induction on $n$, with the case $n = 2$ holding by the rules\Rbetainlrone,\Rbetainlrtwo and\Retainlr.
\end{proof}

\begin{lemma}
\label{lm:assertpi}
\begin{align*}
\instr_{p_i}(x) & = \case_{j=1}^n \instr_{(p_1, \ldots, p_n)}(x) \of \nin{j}{n}{x} \mapsto
\begin{cases}
\inl{x} & \text{if } i = j \\
\inr{x} & \text{if } i \neq j
\end{cases} \\
\assert_{p_i}(x) & = \case_{j=1}^n \instr_{(p_1, \ldots, p_n)}(x) \of \nin{j}{n}{x} \mapsto \begin{cases}
\return x & \text{if } i = j \\
\fail & \text{if } i \neq j
\end{cases}
\end{align*}
\end{lemma}

\begin{proof}
The first formula holds because $\inlprop{}$ maps the right-hand side to $\intest{i}{\instr_{(p_1, \ldots, p_n)}(x)} = p_i$,
and $\nabla$ mapst the right-hand side to $x$.
The second formula follows immediately from the first.
\end{proof}

\begin{lemma}
  \item If $(p,q)$ is a 2-test, then $q = p^\bot$, and $\mathsf{instr}_{(p,q)}(t) = \mathsf{instr}_{p}(t)$.
\end{lemma}

\begin{proof}
  If $(p,q)$ is a 2-test then $p \ovee q = \top$ and so $q = p^\bot$ by Proposition \ref{prop:logic}.\ref{prop:ortho}.  Then
$\mathsf{instr}_{(p,q)}(t) = \mathsf{instr}_p(t)$ by\Retainstr, since $\inlprop{\mathsf{instr}_{(p,q)}(x)} = \langle p ? \rangle \top \ovee \langle q ? \rangle \bot = p$
and $\nabla(\mathsf{instr}_{(p,q)}(x)) = x$.
\end{proof}



We can now define the program that divides into $n$ branches depending on the outcome of an $n$-test:

\begin{definition}
\label{df:measure}
Given $x : A \vdash p_1(x) \ovee \cdots \ovee p_n(x) = \top : \mathbf{2}$, define
\begin{align*}
x : A & \vdash \meas\ p_1(x) \mapsto t_1(x) \mid \cdots \mid p_n(x) \mapsto t_n(x) \\
& \eqdef \case \mathsf{instr}_{(p_1, \ldots, p_n)}(x) \of \inn_1(x) \mapsto t_1(x) \mid \cdots \mid \inn_n(x) \mapsto t_n(x)
\end{align*}
\end{definition}


\begin{lemma}
\label{lm:measure}
The $\meas$ construction satisfies the following laws.
  \begin{enumerate}
  \item \label{lm:measuretop} $(\meas\ \top \mapsto t) = t$
  \item \label{lm:measurebot} $(\meas\ p_1 \mapsto t_1 \mid \cdots \mid p_n \mapsto t_n \mid \bot \mapsto t_{n+1}) = (\meas\ p_1 \mapsto t_1 \mid \cdots \mid p_n \mapsto t_n)$
  \item \label{lm:measureand} $(\meas_i\ p_i \mapsto \meas_j\ q_{ij} \mapsto t_{ij}) = (\meas_{i,j}\ p_i \andthen q_{ij} \mapsto t_{ij})$
  \item \label{lm:measureperm} For any permutation $\pi$ of $\{1, \ldots, n\}$, $\meas_i\ p_i \mapsto t_i = \meas_i\ p_{\pi(i)} \mapsto t_{\pi(i)}$.
  \item \label{lm:measureor} If $t_n = t_{n+1}$ then \\ $\meas_{i=1}^n p_i \mapsto t_i = \meas\ p_1 \mapsto t_1 \mid \cdots \mid p_{n-1} \mapsto t_{n-1} \mid p_n \ovee p_{n+1} \mapsto t_n$.
  \end{enumerate}
\end{lemma}

\begin{proof}
  \begin{enumerate}
  \item
$    \begin{aligned}[t]
      \meas \top \mapsto t(x) & \eqdef \case \instr_{(\top)}(x) \of \nin{1}{1}{x} \mapsto t(x) \\
& = t(\instr_{(\top)}(x))
    \end{aligned}$.

So it suffices to prove $\instr_{(\top)}(s) = s$.
This holds by the uniqueness of Lemma \ref{lm:instrn}, since we have $\intest{1}{x} = \top$ and $\nabla(x) = x$.
\item 
It suffices to prove $\instr_{(p_1, \ldots, p_n, \bot)}(x) = \case_{i=1}^n \instr_{(p_1, \ldots, p_n)}(x) \of \nin{i}{n}{x} \mapsto \nin{i}{n+1}{x}$.
Let $R$ denote the right-hand side.  Then
\begin{align*}
\intest{i}{R} & = \intest{i}{\instr_{(p_1, \ldots, p_n)}(x)} = p_i \\
  \nabla(R) & = \case_{i=1}^n \instr_{(p_1, \ldots, p_n)}(x) \of \nin{i}{n}{x} \mapsto x \\
& = \nabla(\instr_{(p_1, \ldots, p_n)}(x)) = x
\end{align*}
\item 
Let us write $\nin{i,j}{}{}$ ($1 \leq i \leq m$, $1 \leq j \leq n_i$) for the constructors of $(n_1 + \cdots + n_m) \cdot A$,
and $\intest{i,j}{}$ for the corresponding predicates.

It suffices to prove that
\[ \instr_{(p_i \andthen q_{ij})_{i,j}}(x) = \case_{i=1}^m\ \instr_{\vec{p}}(x) \of \nin{i}{m}{x} \mapsto
\case_{j=1}^{n_1}\ \instr_{\vec{q_i}}(x) \of \nin{j}{n_i}{x} \mapsto \nin{i,j}{}{x} \enspace . \]
  Let $R$ denote the right-hand side.  We have
\begin{align*}
\intest{i,j}{R} & = \case_{i'=1}^m\ \instr_{\vec{p}}(x) \of \nin{i'}{m}{x} \mapsto \begin{cases}
\intest{j}{\instr_{\vec{q_i}}(x)} & \text{if } i = i' \\
\bot & \text{if } i \neq \i'
\end{cases} \\
& = \case_{i'=1}^m\ \instr_{\vec{p}}(x) \of \nin{i'}{m}{x} \mapsto \begin{cases}
q_{ij} & \text{if } i = i' \\
\bot & \text{if } i \neq i'
\end{cases} \\
& = \doo{x}{\left( \case_{i'=1}^m\ \instr_{\vec{p}}(x) \of \nin{i'}{m}{x} \mapsto \begin{cases}
\return x & \text{if } i = i' \\
\fail & \text{if } i \neq i'
\end{cases} \right)}{q_{ij}} \\
& = \doo{x}{\assert_{p_i}(x)}{q_{ij}} \\
& \qquad \qquad \text{(by Lemma \ref{lm:assertpi})} \\
& = p_i \andthen q_{ij}
\end{align*}
and
\begin{align*}
\nabla(R) & = \case_{i=1}^m\ \instr_{\vec{p}}(x) \of \nin{i}{m}{x} \mapsto \nabla(\instr_{\vec{q_i}}(x)) \\
& = \case_{i=1}^m\ \instr_{\vec{p}}(x) \of \nin{i}{m}{x} \mapsto x = \nabla(\instr_{\vec{p}}(x)) = x
\end{align*}
\item 
It is sufficient to prove that 
\[ \instr_{(p_1, \ldots, p_n)}(x) = \case_{i=1}^n \instr_{(p_{\pi(1)}, \ldots, p_{\pi(n)})}(x) \of \nin{i}{n}{x} \mapsto \nin{\pi^{-1}(i)}{n}{x} 
\enspace . \]
Let $R$ denote the right-hand side.  We have
\begin{align*}
\intest{i}{R} & = \intest{\pi^{-1}(i)}{\instr_{(p_{\pi(1)}, \ldots, p_{\pi(n)})}(x)} = p_i \\
  \nabla(R) & = \nabla(\instr_{(p_{\pi(1)}, \ldots, p_{\pi(n)})}(x)) = x
\end{align*}
\item
It suffices to prove $\instr_{(p_1, \ldots, p_{n-1}, p_n \ovee p_{n+1})} = \case_{i=1}^{n+1} \instr_{\vec{p}}(x) \of \nin{i}{n}{x} \mapsto \begin{cases}
\nin{i}{n}{x} & \text{if } i < n \\ \nin{i}{n}{x} & \text{if } i \geq n \end{cases}$.  Let $R$ denote the right-hand side.  We have, for $i < n$:
\begin{align*}
\intest{i}{R} & = \intest{i}{\instr_{\vec{p}}(x)} = p_i
  \intest{n}{R} & = \rhd_{n,n+1}(\ind{\instr_{\vec{p}}(x)}) \\
& = \intest{n}{\instr_{\vec{p}}(x)} \ovee \intest{n+1}{\instr_{\vec{p}}(x)} = p_n \ovee p_{n+1}
\nabla(R) & = x \enspace .
\end{align*}
  \end{enumerate}
\end{proof}



Let $x : A \vdash p : \mathbf{2}$ and $\Gamma, x : A \vdash s,t : B$.  We define
\[ \cond{p}{s}{t} \eqdef \meas\ p \mapsto s \mid p^\bot \mapsto t : B \enspace . \]

\begin{lemma}
\label{lm:measuretwo}
  \begin{enumerate}
  \item If $x : A \vdash p_1 \ovee \cdots \ovee p_n = \top : \mathbf{2}$ and $x : A \vdash q_1, \ldots, q_n : \mathbf{2}$, then
\[ (\mathsf{measure}\ p_1 \mapsto q_1 \mid \cdots \mid p_n \mapsto q_n) = p_1 \andthen q_1 \ovee \cdots \ovee p_n \andthen q_n \enspace . \]
  \item Let $x : A \vdash p : \mathbf{2}$ and $\Gamma \vdash q,r : B$ where $x \notin \Gamma$.  Then $\cond{p}{q}{r} = \pcase{p}{\_}{q}{\_}{r} : B$. \label{lm:measurecond}
  \item\label{lm:measuretwo'} Let $x : A \vdash p : \mathbf{2}$.  Then $x : A \vdash \cond{p}{\top}{\bot} = p : \mathbf{2}$.
  \end{enumerate}
\end{lemma}

\begin{proof}
  \begin{enumerate}
  \item
    Immediate from Lemma \ref{lm:instrn}.
  \item We have
    \begin{align*}
      \meas\ p \mapsto q \mid p^\bot \mapsto r & \eqdef \case \instr_{\lambda x p}(x) \of \inl{\_} \mapsto q \mid \inr{\_} \mapsto r \\
& = \case \inlprop{\instr_{\lambda x p}(x)} \of \inl{\_} \mapsto q \mid \inr{\_} \mapsto r \\
& = \case p \of \inl{\_} \mapsto q \mid \inr{\_} \mapsto r
    \end{align*}
  \item $\cond{p}{\top}{\bot} = \pcase{p}{\_}{\top}{\_}{\bot} = p$ by\Retaplus.
  \end{enumerate}
\end{proof}

\subsection{Scalars}
\label{sec:scalars}

From the rules given in Figure \ref{fig:equations}, the usual algebra of the rational interval from 0 to 1 follows.

\begin{lemma}
  If $p / q = m / n$ as rational numbers, then $\vdash p \cdot (1 / q) = m \cdot (1 / n) : \mathbf{2}$.
\end{lemma}

\begin{proof}
We first prove that $\vdash a \cdot (1 / a b) = 1 / b : \mathbf{2}$ for all $a$, $b$.  This holds because $ab \cdot (1 / ab) = \top$ by\Rntimesoneovern,
hence $a \cdot (1 / ab) = 1/b$ by\Rdivide.

Hence we have $p \cdot (1 / q) = pn \cdot (1 / nq) = qm \cdot (1 / n q) = m \cdot (1 /n)$.
\end{proof}

Recall that within $\COMET$, we are writing $m / n$ for the term $m \cdot (1 / n)$.

\begin{lemma}
\label{lm:rational}
Let $q$ and $r$ be rational numbers in $[0,1]$.
\begin{enumerate}
\item If $q \leq r$ in the usual ordering, then $\vdash q \leq r : \mathbf{2}$.
\item $\vdash q \ovee r : \mathbf{2}$ iff $q + r \leq 1$, in which case $\Gamma \vdash q \ovee r = q + r : \mathbf{2}$.
\item $\vdash q \andthen r = qr : \mathbf{2}$.
\end{enumerate}
\end{lemma}

\begin{proof}
  By the previous lemma, we may assume $q$ and $r$ have the same denominator. Let $q = a / n$ and $r = b / n$.
  \begin{enumerate}
  \item We have $a \leq b$, hence $\vdash a \cdot (1 / n) \leq b \cdot (1 / n) : \mathbf{2}$ by Lemma \ref{lm:ordering}.\ref{lm:leqovee}.
  \item If $q + r \leq 1$ then $\vdash a \cdot (1 / n) \ovee b \cdot (1 / n) = (a + b) \cdot (1 / n) : \mathbf{2}$ by Associativity.

For the converse, suppose $\vdash q \ovee r : \mathbf{2}$, so $\vdash (a + b) \cdot (1 / n) : \mathbf{2}$, and suppose for a contradiction $q + r > 1$.  Then we have
\[ \vdash \top \ovee (a + b - n) \cdot (1 / n) : \mathbf{2} \]
and so $\vdash (1 / n) = 0 : \mathbf{2}$ by the Zero-One Law, hence $\vdash \top = n \cdot (1 / n) = n \cdot 0 = \bot : \mathbf{2}$.  This contradicts Corollary \ref{cor:consistency}.
\item We first prove $(1 / a) \andthen (1 / b) = 1 / ab : \mathbf{2}$.  This holds because $ab \cdot (1 / a) \andthen (1 / b) = (a \cdot (1 / a)) \andthen (b \cdot (1 / b)) = \top \andthen \top = \top$.

Now we have, $(m / n) \andthen (p / q) = mp \cdot ((1 / n) \andthen (1 /q)) = mp \cdot (1 / nq)$ as required.
  \end{enumerate}
\end{proof}

\subsection{Normalisation}

The following lemma gives us a rule that allows us to calculate the normalised form of a substate in many cases, including the examples in Section \ref{section:examples}.

\begin{lemma}
Let $\vdash t : A + 1$, $\vdash p_1 \ovee \cdots \ovee p_n = \top : \mathbf{2}$, and $\vdash q : \mathbf{2}$.  Let $\vdash s_1, \ldots, s_n : A$.  Suppose $\vdash 1 / m \leq q : \mathbf{2}$.  If
\[ \vdash t = \meas\ p_1 \andthen q \mapsto \return{s_1} \mid \cdots \mid p_n \andthen q \mapsto \return{s_n} \mid q^\bot \mapsto \fail : A + 1 \]
then
\[ \vdash \norm{t} = \meas\ p_1 \mapsto s_1 \mid \cdots \mid p_n \mapsto s_n : A \]
\end{lemma}

\begin{proof}
Let $\rho \eqdef \meas_{i=1}^n p_i \mapsto s_i$.
  By the rule\Retanorm, it is sufficient to prove that $t  = \doo{\_}{t}{\return{\rho}}$.
We have
\begin{align*}
  \doo{\_}{t}{\return{\rho}}
& = \meas\ p_1 \andthen q \mapsto \return{\rho} \mid \cdots \mid p_n \andthen q \mapsto \return{\rho} \mid q^\bot \mapsto \fail \\
& = \meas\ (p_1 \ovee \cdots \ovee p_n) \andthen q \mapsto \return{\rho} \mid q^\bot \mapsto \fail \\
& = \meas\ q \mapsto \return{\rho} \mid q^\bot \mapsto \fail \\
& = \meas\ q \mapsto \meas_{i=1}^n p_i \mapsto \return{s_i} \mid q^\bot \mapsto \fail \\
& = \meas_{i=1}^n\ q \andthen p_i \mapsto \return{s_i} \mid q^\bot \mapsto \fail \\
& = t
\end{align*}
(We used the commutativity of $\andthen$ in the last step.)
\end{proof}

\begin{corollary}
\label{cor:normmeasure}
Let $\alpha_1$, \ldots, $\alpha_n$, $\beta$ be rational numbers that sum to 1, with $\beta \neq 1$.  If
\[ \vdash t = \meas\ \alpha_1 \mapsto \return{s_1} \mid \cdots \mid \alpha_n \mapsto \return{s_n} \mid \beta \mapsto \fail : A + 1 \]
then
\[ \vdash \norm{t} = \meas\ \alpha_1 / (\alpha_1 + \cdots + \alpha_n) \mapsto s_1 \mid \cdots \mid \alpha_n / (\alpha_1 + \cdots + \alpha_n) \mapsto s_n : A \]
\end{corollary}

\section{Semantics}
\label{section:semantics}

The terms of $\COMET$ are intended to represent probabilistic programs.
We show how to give semantics to our system in three different ways: using discrete and continuous probability distributions, and
simple set-theoretic semantics for deterministic computation.

\subsection{Discrete Probabilistic Computation}
\label{section:dpc}

We give an interpretation that assigns, to each term, a discrete probability distribution over its output type.

\begin{definition}
Let $A$ be a set.
\begin{itemize}
\item
The \emph{support} of a function $\phi : A \rightarrow [0,1]$ is $\supp \phi = \{ a \in A : \phi(a) \neq 0 \}$.
\item 
A \emph{(discrete) probability distribution} over $A$ is a function $\phi : A \rightarrow \phi$ with finite support
such that $\sum_{a \in A} \phi(a) = 1$.
\item 
Let $\mathcal{D} A$ be the set of all probability distributions on $A$.
\end{itemize}
\end{definition}

We shall interpret every type $A$ as a set
$\brackets{A}$.  Assume we are given a set $\brackets{\mathbf{C}}$ for each type constant $\mathbf{C}$.
Define a set $\brackets{A}$ for each type $A$ thus:
\[ \brackets{0} = \emptyset \qquad \brackets{1} = \{ * \} \qquad \brackets{A + B} = \brackets{A} \uplus \brackets{B} \qquad \brackets{A \sotimes B} = \brackets{A} \times \brackets{B} \]
where $A \uplus B = \{ a_1 : a \in A \} \cup \{ b_2 : b \in B \}$.  We extend this to contexts by defining $\brackets{x_1 : A_1, \ldots, x_n : A_n} = \brackets{A_1} \times \cdots \times \brackets{A_n}$.

Now, to every term $x_1 : A_1, \ldots, x_n : A_n \vdash t : B$, we assign a function
$\brackets{t} : \brackets{A_1} \times \cdots \times \brackets{A_n} \rightarrow \mathcal{D} \brackets{B}$.
The value $\brackets{t}(a_1, \ldots, a_n)(b) \in [0,1]$ will be written as $P(t(a_1, \ldots, a_n) = b)$, and should be thought of as the probability
that $b$ will be the output if $a_1$, \ldots, $a_n$ are the inputs.

\begin{figure}
\begin{mdframed}
\begin{multicols}{2}
$$\begin{aligned}
P(x_i(\vec{a}) = b) & = \begin{cases}
1 \text{ if } b = a_i \\
 0 \text{ if } b \neq a_i 
\end{cases} \\ \midrule
P(*(\vec{a}) = *) & = 1 \\ \midrule
\multicolumn{2}{l}{$P((s \sotimes t)(\vec{g}, \vec{d}) = (a,b))$} \\
& = P(s(\vec{g}) = a) P(t(\vec{d}) = b) \\ \midrule
P((\magic{t})(\vec{g}) = a) & = 0 \\ \midrule
P(\inl{t}(\vec{g}) = a_1) & = P(t(\vec{g}) = a) \\
P(\inl{t}(\vec{g}) = b_2) & = 0 \\ \midrule
P(\inr{t}(\vec{g}) = a_1) & = 0 \\
P(\inr{t}(\vec{g}) = b_2) & = P(t(\vec{g}) = b) \\ \midrule
P(\inlr{s}{t}(\vec{g}) = a_1) & = P(s(\vec{g}) = a_1) \\
P(\inlr{s}{t}(\vec{g}) = b_2) & = P(t(\vec{g}) = b_1)
\end{aligned}$$

$\begin{aligned}
\multicolumn{2}{l}{$P(\lft{t}(\vec{g}) = a) = P(t(\vec{g}) = a_1)$} \\ \midrule
\multicolumn{2}{l}{$P(\instr_{\lambda x t}(s)(\vec{g}) = a_i)$} \\ & = P(s(\vec{g}) = a) P(t(a) = i) \\ \midrule
\multicolumn{2}{l}{$P(1 / n(\vec{g}) = \top) = 1 / n$} \\
\multicolumn{2}{l}{$P(1 / n(\vec{g}) = \bot) = (n - 1) / n$} \\ \midrule
\multicolumn{2}{l}{$P(\norm{t}(\vec{g}) = a)$} \\ & \ = P(t(\vec{g}) = a_1) / (1 - P(t(\vec{g}) = *_2)) \\ \midrule
\multicolumn{2}{l}{$P((s \ovee t)(\vec{g}) = a_1)$} \\ & \ = P(s(\vec{g}) = a_1) + P(t(\vec{g}) = a_1) \\
\multicolumn{2}{l}{$P((s \ovee t)(\vec{g}) = *_2)$} \\ & \ = P(s(\vec{g}) = *_2) + P(t(\vec{g}) = *_2) - 1
\end{aligned}$
\end{multicols}
$\begin{aligned}
& P((\plet{x}{y}{s}{t})(\vec{g},\vec{d}) = c) = \sum_a \sum_b P(s(\vec{g}) = (a,b)) P(t(\vec{d},a,b) = c) \\ \midrule
& P(\pcase{r}{x}{s}{y}{t}(\vec{g},\vec{d}) = c) \\
& = \sum_a P(r(\vec{g}) = a_1) P(s(\vec{d}, a) = c) + 
\sum_b P(r(\vec{g}) = b_2) P(t(\vec{d}, b) = c)
\end{aligned}$
\end{mdframed}
\end{figure}

The sums involved here are all well-defined because, for all $t$ and $\vec{g}$, the function $P(t(\vec{g}) = -)$ has finite support.

\begin{lemma}
Let $\Gamma \vdash s : A$ and $\Delta, x : A \vdash t : B$, so that $\Gamma, \Delta \vdash t[x:=s] : B$.  Then
\[ P(t[x:=s](\vec{g}, \vec{d}) = b) = \sum_{a \in \brackets{A}} P(s(\vec{g}) = a) P(t(\vec{d},a) = b) \]
\end{lemma}

\begin{proof}
The proof is by induction on $t$.  We do here the case where $t \equiv x$:
\[ P(x[x:=s](\vec{g}) = b) = P(s(\vec{g}) = b) \]
and
\[ \sum_a P(s(\vec{g}) = a) P(x(a) = b) = P(s(\vec{g}) = b) \]
since $P(x(a) = b)$ is 0 if $a \neq b$ and 1 if $a = b$.
\end{proof}

\begin{theorem}[Soundness]
  \begin{enumerate}
  \item If $\Gamma \vdash t : A$ is derivable, then for all $\vec{g} \in \brackets{\Gamma}$, we have $P(t(\vec{g}) = -)$ is a
probability distribution on $\brackets{A}$.
\item If $\Gamma \vdash s = t : A$, then $P(s(\vec{g}) = a) = P(t(\vec{g}) = a)$.
  \end{enumerate}
\end{theorem}

\begin{proof}
  The proof is by induction on derivations.  We do here the case of the rule\Rinstrtest:
  \begin{align*}
&    P((\case_i\ \instr_{\lambda x t}(s) \of \nin{i}{n}{\_} \mapsto i)(\vec{g}) = i) \\
& = \sum_{j = 1}^n \sum_{a \in \brackets{A}} P(\instr_{\lambda x t}(s)(\vec{g}) = a_j) P(\nin{i}{n}{*}() = *_j) \\
& = \sum_{a \in \brackets{A}} P(\instr_{\lambda x t}(s)(\vec{g}) = a_i) \\
& = \sum_{a \in \brackets{A}} P(s(\vec{g}) = a) P(t(a) = i) \\
& = P(t[x:=s](\vec{g}) = i)
  \end{align*}
by the lemma.
\end{proof}

\begin{corollary}
  If $\Gamma \vdash s \leq t : A + 1$ then $P(s(\vec{g}) = a) \leq P(t(\vec{g}) = a)$ for all $\vec{g}$, $a$.
\end{corollary}

As a corollary, we know that $\COMET$ is non-degenerate:

\begin{corollary}
\label{cor:consistency}
  Not every judgement is derivable; in particular, the judgement $\vdash \top = \bot : \mathbf{2}$ is not derivable.
\end{corollary}

With these definitions, we can calculate the semantics of each of our defined constructions.  For example,
the semantics of $\mathsf{assert}$ are given by
\[ P(\assert_{\lambda x p}(t)(\vec{g}) = a_1) = P(t(\vec{g}) = a)P(p(a) = \top) \]
\[ P(\assert_{\lambda x p}(t)(\vec{g}) = *_2) = \sum_a P(t(\vec{g}) = a) P(p(a) = \bot) \]

\subsection{Alternative Semantics}

It is also possible to give semantics to $\COMET$ using continuous probabilities.  We assign a measurable space $\brackets{A}$ to every type $A$.  Each term then gives a measurable function $\brackets{A_1} \times \cdots \times \brackets{A_n} \rightarrow \mathcal{G} \brackets{B}$, where $\mathcal{G} X$ is the space of all probability distributions over the measurable space $X$.  ($\mathcal{G}$ here is the \emph{Giry monad} \cite{Jacobs13a}.)

If we remove the constants $1 / n$ from the system, we can give \emph{deterministic} semantics to the subsystem, in which we assign a set to every type, and a function $\brackets{A_1} \times \cdots \times \brackets{A_n} \rightarrow \brackets{B}$. 

More generally, we can give an interpretation of $\COMET$ in any \emph{commutative monoidal effectus with normalisation}
in which there exists a scalar $s$ such that $n \cdot s = 1$ for all positive integers $n$ \cite{Cho}.  The discrete and continuous semantics we have described are two instances of this interpretation.

\section{Conclusion}

The system $\COMET$ allows for the specification of probabilistic programs and reasoning about their properties, both within the same syntax.  

There are several avenues for further work and research.
\begin{itemize}
\item The type theory that we describe can be interpreted both in
  discrete and in continuous probabilistic models, that is, both in
  the Kleisli category $\Kl(\Dst)$ of the distribution monad $\Dst$
  and in the Kleisli category $\Kl(\Giry)$ of the Giry monad $\Giry$.
  On a finite type each distribution is discrete. The discrete semantics were exploited in
  the current paper in the examples in Section~\ref{section:examples}.
  In a follow-up version we intend to elaborate also continuous
  examples.

\item The normalisation and conditioning that we use in this paper can
  in principle also be used in a quantum context, using the
  appropriate (non-side-effect free) assert maps that one has
  there. This will give a form of Bayesian quantum theory, as also
  explored in~\cite{LeiferS13}.

\item A further ambitious follow-up project is to develop tool support
  for $\COMET$, so that the computations that we carry out here by
  hand can be automated. This will provide a formal language for
  Bayesian inference.
\end{itemize}

\subparagraph*{Acknowledgements}

Thanks to Kenta Cho for discussion and suggestions during the writing of this paper, and very detailed proofreading.  Thanks to Bas Westerbaan for discussions about effectus theory.

\bibliography{probable}

\appendix

\section{Formal Presentation of $\COMET$}
\label{section:rules}

The full set of rules of deduction for $\COMET$ are given below.

\subsection{Structural Rules}
\label{section:structural}
$$ \Texch \qquad \Tvar $$

The exchange rule says that the order of the variables in the context does not matter. This holds
for all types of judgements J on the right hand side of the turnstile.   The weakening rule is admissible (see Lemma \ref{lm:meta}.\ref{lm:weak}), and says
that one may add (unused) assumptions to the context. 

However, we do \emph{not} have the contraction rule in our type theory.  In particular, the judgement $x : A \vdash x \otimes x : A \otimes A$ is \emph{not} derivable.
Thus, in our probabilistic settings, information may be discarded, but cannot be duplicated.  

$$ \Tref \; \Tsym \; \Ttrans $$

These rules simply ensure that the judgement equality is an equivalence relation.

\subsection{The Singleton Type}

$$ \Tunit \quad \Tetaone $$

These ensure that the type $1$ is a type with only one object up to equality.

\subsection{Tensor Product}

$$ \Tpair \; \Tlett $$
$$ \Tpaireq $$
$$ \Tleteq $$

Notice that in rule\Rpair the contexts $\Gamma$ and
$\Delta$ of the two terms $s$, $t$ are put together in the
conclusion. Thus, the tensor $s \sotimes t$ on terms is a form of
parallel composition. This is a so-called \emph{introduction rule} for the
tensor type, since it tells us how to produce terms in a tensor type
$A\otimes B$ on the right hand side of the turnstile $\vdash$. The
rule\Rlett is an \emph{elimination rule} since it tells us how to use terms
of tensor type.

$$ \Tbeta $$
$$ \Teta $$

Rule\Rbeta tells how a let
term should decompose a term $r \sotimes s$, namely by simultaneously
substituting $r$ for $x$ and $s$ for $y$ in as described in the term
$t[x:=r,y:=s]$.  Rule\Reta is its dual, and says that decomposing an object then
immediately recomposing it does nothing.

$$ \Tletlet $$
$$ \Tletpair $$

\noindent Our final set of rules are so-called commuting conversion
rules described above. They regulate the proper interaction between
the term constructs let, case and $\sotimes$.  It looks like several interactions are missing here (a $\lett$ on the right
of a tensor, a $\lett$ inside a $\case$, etc.), but in fact, the rules for all the other cases can be derived from these four, as we show in
Lemma \ref{lm:sub}.\ref{lm:letsub}.

\subsection{Empty Type}

$$ \Tmagic \quad \Tetazero $$

The rule\Rmagic
says that from an inhabitant $M:0$ we can produce an inhabitant
$\magic{M}$ in any type $A$. Intuitively, this says `If the empty type is inhabited, then every type is inhabited', which is vacuously true.
And\Retazero says that vacuously, if the empty type $0$ is inhabited, then all terms of any type are equal.

\subsection{Binary Coproducts}
\label{section:coproducts}
$$ \Tinl \quad \Tinr $$
$$ \Tinleq \quad \Tinreq $$
$$\Tcase $$
$$\Tcaseeq $$

For the coproduct type $A+B$ we have two introduction rules\Rinl
and\Rinr which produce terms $\inl{s}, \inr{t} : A+B$, coming from
$s:A$ and $t:B$.  These operations $\inl{-}$ and $\inr{-}$ are often
called \emph{coprojections} or \emph{injections}.

The associated elimination
rule\Rcase produces a term that uses a term $r:A+B$ by distinguishing
whether or not $r$ is of the form $\inl{-}$ or $\inr{-}$. In the first
case the outcome of $r$ is used in term $s$, and in the second case in
term $t$.

$$ \Tbetaplusone $$
$$ \Tbetaplustwo $$
$$ \Tetaplus $$

There are two $\beta$-conversions\Rbetaplusone and\Rbetaplustwo
for the coproduct type, describing how a $\mathsf{case}$ term should handle a
term of form $\inl{r}$ or $\inr{r}$. Again this this
done via the expected substitution, using the appropriate variable
($x$ or $y$).

In rule\Retaplus, if the decomposition of $t$ into $\inl{-}$ and $\inr{-}$
is then immediately reconstituted, then the input is unchanged.
$$ \Tcasecase $$
$$ \Tcasepair $$
$$ \Tletcase $$
These rules for commuting conversions show how the eliminators for $\otimes$ and $+$ interact.  Again,
the other cases can be derived from the primitive rules given here (Lemma \ref{lm:sub}).

\subsection{Partial Pairing}
\label{section:effectus}

We now come to the constructions that are new to our type theory.  These possess a feature that is unique to this type theory:
we allow typing judgements (of the form $t : A$) to depend on equality judgements (of the form $s = t : A$).

$$ \Tinlr $$
$$ \Tinlreq $$

The term $\inlr{s}{t}$ can be understood in this way.  Consider a term $\Gamma \vdash t : A + 1$ as a partial computation:
it may output a value of type $A$, or it may diverge (if it reduces to $\inr{*}$.)  If the judgement $s \downarrow = t \uparrow$ holds,
then we know that exactly one of the computations $s$ and $t$ will terminate on any input.  The term $\inlr{s}{t}$ intuitively denotes the following computation:
given an input, decide which of $s$ or $t$ will terminate.  If $s$ will terminate, run $s$; otherwise, run $t$.


We have the following $\beta$- and $\eta$-rules for the $\inlrn$ construction:

$$ \Tbetainlrone $$
$$ \Tbetainlrtwo $$
$$ \Tetainlr $$

\subsection{The $\lft{}$ Construction}

$$ \Tleft $$
$$ \Tlefteq $$

The term $\lft{t}$ should be understood as follows: if we have a term $t : A + B$ and a `proof' that $t = \inl{s}$ for some term $s : A$, then
$\lft{t}$ is that term $s$.  The computation rules for this construction are:

$$ \Tbetaleft \Tetaleft $$

\subsection{Joint Monicity Condition}
\label{section:JM}

We need the following rule for technical reasons.  It corresponds to the condition that the two maps $\rhd_!$ and $\rhd_2$ from $A + A$ to $A$ are jointly monic
in the partial form of the effectus (see \cite{Jacobs14} Assumption 1 or \cite{Cho} Lemma 49.4).

\TTJMprime

It is used in the proof of the associativity of $\ovee$ (Lemma \ref{lm:ordering}.\ref{lm:assoc}).

\subsection{Instruments}
\label{section:instruments}

The \emph{instrument} map $\instr_{\lambda x t}(s)$ should be understood as follows: it denotes the value $\nin{i}{n}{s}$ if $t[x:=s]$ returns the value $i : \mathbf{n}$.

If we were allowed to simply duplicate data, we could have defined $\measure{x}{p}{t}$ to be $\pcase{[t/x]p}{\_}{\inl{t}}{\_}{\inr{t}}$.  This cannot be done in our system, as it would involve duplicating the variables in $t$.

The computation rules for this construction are as follows.

$$ \Tinstr \quad \Tnablainstr $$
$$ \Tinstrtest $$
$$ \Tetainstr $$
$$ \Tinstreq $$

We also introduce the following rule, which ensures that the sequential product $\andthen$ is commutative.

\TTcomm

\subsection{Scalar Constants}

For any natural number $n \geq 2$, we have the following rules.

$$ \Toneovern \; \Tntimesoneovern $$
$$ \Tdivide \; \Tboundmn $$
$$ \Trhdoneboundmn $$
$$ \Trhdtwoboundmnprime $$

These ensure that $1 / n$ is the unique scalar whose sum with itself $n$ times is $\top$.
The term $b_{mn}$ is required to ensure that the term $1 / n \ovee \cdots \ovee 1 / n$ is well-typed.

\subsection{Normalisation}

Finally, we have these rules for normalisation.

$$ \Tnorm \; \Tbetanorm $$
$$ \Tetanorm $$

These ensure that, if $t$ is a non-zero state in $A + 1$, then $\rho$ is the unique state in $A$ such that
$t = \doo{\_}{t}{\return{\rho}}$.  

\section{Proof of Associativity}
\label{section:associativity}

\begin{theorem}
If $\Gamma \vdash (r \ovee s) \ovee t : A + 1$, then $\Gamma \vdash r \ovee (s \ovee t) : A + 1$ and $\Gamma \vdash r \ovee (s \ovee t) = (r \ovee s) \ovee t : A + 1$.
\end{theorem}

(Note: this proof follows the proofs that $\ovee$ is associative in an effectus, found in \cite{Jacobs14} Proposition 12 or \cite{Cho} Proposition 13.)
 
\begin{proof}
Let $b$ be a bound for $r \ovee s$ and $c$ a bound for $(r \ovee s) \ovee t$, so that
\begin{align}
b \goesto \rhd_1 & = r \label{eq:axiom1} \\
b \goesto \rhd_2 & = s \label{eq:axiom2} \\
\doo{x}{b}{\return{\nabla(x)}} & = r \ovee s \label{eq:axiom3} \\
c \goesto \rhd_1 & = r \ovee s \label{eq:axiom4} \\
c \goesto \rhd_2 & = t \label{eq:axiom5} \\
\doo{x}{c}{\return{\nabla(x)}} & = (r \ovee s) \ovee t \label{eq:axiom6}
\end{align}

Define $d : (A + 1) + 1$ by
\[ d = \case c \of \inl{\inl{\_}} \mapsto \fail \mid \inl{\inr{x}} \mapsto \return{\inl{x}} \mid \inr{\_} \mapsto \return{\inr{*}} \]
We wish to form the term $\inlr{b}{d}$.  To do this, we must prove
$\dom{b} = \ker{d}$.  We do this by proving both are equal to $\dom{(r
  \ovee s)}$.

We have
\begin{align*}
\dom{(r \ovee s)} & = \dom{(\doo{x}{b}{\return{\nabla(x)}})} = \doo{x}{b}{\dom{(\return{\nabla(x)})}} = \doo{x}{b}{\top} = \dom{b}
\end{align*}
and
\begin{align*}
\dom{(r \ovee s)} & = \dom{(\doo{x}{c}{\rhd_1(x)})} = \doo{x}{c}{\dom{(\rhd_1(x))}} = \doo{x}{c}{\inlprop{x}} \\
\ker{d} & = \case c \of \inl{\inl{\_}} \mapsto \top \mid \inl{\inr{\_}} \mapsto \bot \mid \inr{\_} \mapsto \bot \\
& = \doo{x}{c}{\inlprop{y}} \\
\therefore \dom{b} & = \ker{d}
\end{align*}
So, let $e = \inlr{b}{d} : (A + A) + (A + 1)$.  We claim
\begin{align}
  \label{eq:transitivity}
c =  \case e \of & \inl{\inl{a}} \mapsto \return{\inl{a}} \mid \inl{\inr{a}} \mapsto \return{\inl{a}} \mid \\
& \inr{\inl{a}} \mapsto \return{\inr{a}} \mid \inr{\inr{\_}} \mapsto \fail \nonumber
\end{align}

We prove the claim using\RJMprime.  Writing $R$ for the right-hand side of (\ref{eq:transitivity}), we have
\begin{align*}
(RHD \goesto \rhd_1)
& = \doo{x}{\rhd_1(e)}{\return \nabla(x)} = \doo{x}{b}{\return \nabla(x)} = r \ovee s & \text{by (\ref{eq:axiom3})} \\
(c \goesto \rhd_1) & = r \ovee s & \text{by (\ref{eq:axiom4})} \\
(R \goesto \rhd_2)
& = (\doo{x}{\rhd_2(e)}{x}) = (\doo{x}{d}{x}) = (c \goesto \rhd_2)
\end{align*}
and so (\ref{eq:transitivity}) follows by\RJMprime.

Now that the claim (\ref{eq:transitivity}) is proved, we return to the main proof.  Define $e' : (A + A) + 1$ by
\begin{align*}
e' = \case e \of & \inl{\inl{\_}} \mapsto \fail \mid \inl{\inr{a}} \mapsto \return{\inl{a}} \mid \\
&  \inr{\inl{a}} \mapsto \return{\inr{a}} \mid \inr{\inr{\_}} \mapsto \fail
\end{align*}
We claim $e'$ is a bound for $s \ovee t$.  We have
\begin{align*}
(e' \goesto \rhd_1)
& = \case e \of \inl{\inl{\_}} \mapsto \fail \mid \inl{\inr{a}} \mapsto \return{a} \mid \inr{\_} \mapsto \fail \\
& = (\rhd_1(e) \goesto \rhd_2) = (b \goesto \rhd_2) = s & \text{by (\ref{eq:axiom2})} \\
(e' \goesto \rhd_2) & = \case e \of \inl{\_} \mapsto \fail \mid \inr{\inl{a}} \mapsto \return{a} \mid \inr{\inr{\_}} \mapsto \fail \\
& = (\rhd_2(e) \goesto \rhd_1) = (d \goesto \rhd_1) \\
& = \case c \of \inl{\inl{\_}} \mapsto \fail \mid \inl{\inr{x}} \mapsto \return{x} \mid \inr{\_} \mapsto \fail \\
& = (c \goesto \rhd_2) = t & \text{by (\ref{eq:axiom5})}
\end{align*}
and so
\begin{align}
s \ovee t = & \doo{x}{e'}{\return{\nabla(x)}} \label{eq:soveet} \\
= & \case e \of \inl{\inl{\_}} \mapsto \fail \mid \inl{\inr{a}} \mapsto \return{a} \mid \nonumber \\
& \inr{\inl{a}} \mapsto \return{a} \mid \inr{\inr{\_}} \mapsto \fail \nonumber \\
\end{align}
Now, define $e'' : (A + A) + 1$ by
\begin{align*}
  e'' = \case e \of & \inl{\inl{a}} \mapsto \return{\inl{a}} \mid \inl{\inr{a}} \mapsto \return{\inr{a}} \\
& \inr{\inl{a}} \mapsto \return{\inr{a}} \mid \inr{\inr{\_}} \mapsto \fail
\end{align*}
We will prove that $e''$ is a bound for $r \ovee (s \ovee t)$.  We have
\begin{align*}
(e'' \goesto \rhd_1)
& = \case e \of \begin{array}[t]{l}
\inl{\inl{a}} \mapsto \return{a} \\
\mid \inl{\inr{\_}} \mapsto \fail \\
\mid \inr{\inl{\_}} \mapsto \fail \\
\mid \inr{\inr{\_}} \mapsto \fail 
\end{array} \\
& = (\rhd_1(e) \goesto \rhd_1) = (b \goesto \rhd_1) = r & \text{by (\ref{eq:axiom1})} \\
(e'' \goesto \rhd_2) & = \case e \of \begin{array}[t]{l}
 \inl{\inl{\_}} \mapsto \fail \\
\mid \inl{\inr{a}} \mapsto \return{a} \\
\mid \inr{\inl{a}} \mapsto \return{a} \\
\mid \inr{\inr{\_}} \mapsto \fail 
\end{array} \\
& = s \ovee t & \text{by (\ref{eq:soveet})} \\
\doo{x}{e''}{\return{\nabla(x)}} & = \case e \of \begin{array}[t]{l}
 \inn_1(a) \mapsto \return{a} \\
\inn_2(a) \mapsto \return{a} \\
\inn_3(a) \mapsto \return{a} \\
\inn_4(\_) \mapsto \fail 
\end{array} \\
& = \mathsf{do}\ x \leftarrow \case e \of \begin{array}[t]{l} \inn_1(a) \mapsto \return{\inl{a}} \\
\inn_2(a) \mapsto \return{\inl{a}} \\
\inn_3(a) \mapsto \return{\inr{a}} \\
\inn_4(a) \mapsto \fail; \return{\nabla(x)} 
\end{array} \\
& = \doo{x}{c}{\return{\nabla(x)}} & \text{by (\ref{eq:transitivity})} \\
& = (r \ovee s) \ovee t
\end{align*}
Thus, $r \ovee (s \ovee t) = \doo{x}{e''}{\return{\nabla(x)}} = (r \ovee s) \ovee t$.
\end{proof}

\end{document}